\documentclass[11pt,a4paper]{article}

\usepackage[cp1251]{inputenc}
\usepackage[T2A]{fontenc}
\usepackage[russian]{babel}

\usepackage{hyperref}
\usepackage[dvips]{graphicx}
\usepackage{amsmath}
\usepackage{amssymb}
\usepackage{amsxtra}
\usepackage{amsthm}
\usepackage{latexsym}
\usepackage{array}
\usepackage{multirow}
\usepackage{hhline}

\newtheorem{theorem}{Теорема}
\newtheorem{defin}{Определение}
\newtheorem{remark}{Замечание}
\newtheorem{proposition}{Предложение}
\newtheorem{lemma}{Лемма}

\numberwithin{equation}{section}

\newcounter{myta}
\newcommand{\myt}{\refstepcounter{myta}\themyta}

\newcommand {\A} {\mathop{\rm Acc}\nolimits}
\newcommand {\rk} {\mathop{\rm rank}\nolimits}
\newcommand {\lsgn} {\mathop{\rm bsgn}\nolimits}
\newcommand {\ri} {\mathrm{i}}
\newcommand {\mbf}[1] {\mathbf{#1}}

\newcommand {\mF}{\mathcal{F}}
\newcommand {\mB}{\mathcal{B}}
\newcommand {\mC}{\mathcal{K}}

\newcommand {\msim}[1] {\mathop \sim \limits^{#1}}
\newcommand {\fts}[1] {{\small #1}}
\newcommand {\ts}[1] {\textsl{#1}}
\newcommand {\bT}{\mathbf{T}}
\newcommand {\bR}{\mathbb{R}}
\newcommand {\ds}{\displaystyle}
\newcommand {\mstrut}{\vphantom{\bigl(}}

\begin{document}
\begin{center}

\Large{{\bf Топологический анализ и булевы функции. \\I. Методы и
приложения к классическим системам}}

\vspace{5mm}

\normalsize

{\bf М.П.\,Харламов}

\vspace{4mm}

\small

Волгоградская академия государственной службы

Россия, 400131, Волгоград, ул. Гагарина, 8

E-mail: mharlamov@vags.ru
\end{center}

\begin{flushright}
{\it Получено 27 мая 2010 г.}
\end{flushright}

\vspace{3mm}

{ \footnotesize Рассматривается задача полной формализации грубого
топологического анализа интегрируемых гамильтоновых систем при
наличии аналитического решения, в котором как правые части
дифференциальных уравнений для вспомогательных переменных, так и
исходные фазовые переменные выражаются рациональными функциями, а
значит, и полиномами от некоторого набора радикалов, каждый из
которых зависит только от одной переменной. Приведен способ сведения
задач определения допустимых областей констант первых интегралов,
промежутков осцилляции разделенных переменных и количества связных
компонент интегральных многообразий и критических интегральных
поверхностей к алгоритмам обработки таблиц некоторых булевых
вектор-функций и приведения матриц линейных булевых вектор-функций к
каноническому виду. С этой точки зрения рассмотрены топологически
наиболее богатые классические задачи динамики твердого тела. Новые
интегрируемые задачи будут рассмотрены в продолжении.

}

\normalsize

\vspace{3mm} Ключевые слова: алгебраическое разделение переменных,
интегральные многообразия, булевы функции, топологический анализ,
алгоритмы

\vspace{6mm}

\begin{center}

\large

{\bf M.P.~Kharlamov}

{\bf Topological analysis and Boolean functions. I. Methods and
application to classical systems}

\normalsize
\end{center}

\vspace{3mm}

{\footnotesize We aim to completely formalize the rough topological
analysis of integrable Hamiltonian systems admitting analytical
solutions such that the initial phase variables along with the time
derivatives of the auxiliary variables are expressed as rational
functions (in fact, as polynomials) in some set of radicals
depending on one variable each. We suggest a method to define the
admissible regions in the integral constants space, the segments of
oscillation of the separated variables and the number of connected
components of integral manifolds and critical integral surfaces.
This method is based on some algorithms of processing the tables of
some Boolean vector-functions and of reducing the matrices of linear
Boolean vector-functions to some canonical form. From this point of
view we consider here the topologically richest classical problems
of the rigid body dynamics. The article will be continued with the
investigation of some new integrable problems.

}

\vspace{3mm} Keywords: algebraic separation of variables, integral
manifolds, Boolean functions, topological analysis, algorithms

Mathematical Subject Classification 2000: 70E17, 70G40

\newpage

\tableofcontents

\section{Введение}\label{sec1}
Исторически при использовании дифференциальных уравнений первичной
задачей ставилось их интегрирование, построение явных решений.
Открытие явления принципиальной неинтегрируемости сдвинуло акценты в
исследованиях в сторону качественных методов. Однако в последние
десятилетия, благодаря активному развитию алгебраических подходов к
обнаружению явления интегрируемости динамических, и в первую
очередь, гамильтоновых систем,  было открыто много случаев
коммутативной и некоммутативной интегрируемости, значительная часть
которых относится к динамике твердого тела с неподвижной точкой и ее
обобщениям на группы Ли, отличные от $SO(3)$. В связи с этим вновь
отмечается значительный рост интереса математиков к интегрируемым
гамильтоновым системам. Достаточно отметить изданные совсем недавно
монографии оригинального и обзорного характера
\cite{BorMam0,BorMam,BorMam2,ReySemBk,TsiBk}, содержащие в списках
источников более тысячи наименований.

Относительно недавно открыты и интегрируемые системы с числом
степеней свободы три и выше, не имеющие явных симметрий и потому
несводимые в целом к системам с двумя степенями свободы (так
называемые неприводимые системы). Наиболее ярким примером такого
рода служит обобщенный гиростат Ковалевской -- гиростат,
удовлетворяющий динамическим соотношениям случая Ковалевской,
помещенный в двойное силовое поле. Эта система была обнаружена
благодаря усилиям целого ряда математиков -- О.И.\,Богоявленского
\cite{BogRus1,BogRus2,BogEn}, Х.\,Яхья \cite{Yeh1,Yeh2},
И.В.\,Комарова \cite{Komar}. Самый общий результат был получен
А.Г.Рейманом и М.А.\,Семеновым-Тян-Шанским \cite{ReySem,ReySemRus}.
Однако, как выяснилось, никаких явных решений построить в этом
случае не удается. Первый топологический результат -- описание
фазовой топологии системы, указанной Богоявленским, получен
Д.Б.\,Зотьевым \cite{ZotRCD}. Исследования волчка Ковалевской в
двойном поле, основанные на идее стратификации, построения и
изучения критических подсистем, начаты в работе \cite{Odin}. В итоге
это позволило построить для исходной задачи трехмерную
бифуркационную диаграмму \cite{KhRCD1} и приступить к трехмерной
классификации критических точек с целью построения инвариантного
описания фазовой топологии. Критические подсистемы оказались по
своей общности сопоставимыми со всей классической задачей
Ковалевской и обнаружили все черты систем, в которых выражения для
фазовых переменных и уравнения для разделенных переменных удается
получить алгебраическим путем. В настоящей работе мы называем такие
системы алгебраически разрешимыми. При наличие алгебраического
решения количество степеней свободы уже не является принципиальным.
Поэтому разработанный аппарат применим в равной степени к приводимым
и неприводимым системам.

В современной литературе много работ посвящено {\it алгебраически
интегрируемым} системам (algebraic complete integrable systems).
Этот термин не получил точного и однозначного определения -- его
соотносят, в основном, с гамильтоновыми системами, топология которых
явным образом связана с якобианами и многообразиями Прима
алгебраических поверхностей. Его, скорее, следовало бы перевести как
<<{\it алгебраические} вполне интегрируемые {\it системы}>>,
поскольку алгебраичность здесь -- свойство системы, а не способа
интегрирования. Алгебраические интегрируемые системы, в которых
получено фактическое разделение переменных, как правило,
удовлетворяют и требованию алгебраичности выражений исходных фазовых
переменных, хотя в современных работах, в отличие от классических,
именно этому вопросу не всегда уделяется внимание.

Полученные или анонсированные в последнее десятилетие разделения
переменных (см., например, работы
\cite{Sok1,Tsi1,Tsi3,Tsi4,Tsi5,Tsi6}) либо являются комплексными
(уравнения движения не разделяются в вещественной области, выражения
для исходных вещественных фазовых переменных не получены), либо
сводятся к задаче решения алгебраических уравнений высокой степени
(в последнем случае, как, например, в задаче Клебша, такие
свед\'{е}ния зачастую можно найти уже у классиков). Разделение
переменных в настоящем классическом понимании имеется в работе
\cite{Tsi2}, но оно требует вывода полных зависимостей исходных
фазовых переменных, поскольку, как нам представляется, только такие
зависимости (или хотя бы доказательство их существования через
оценки рангов отображений) определяют факт наличия явного решения.
Разделения переменных, связанные с алгебраическими кривыми,
порожденными представлениями Лакса, существование которых
теоретически доказывается, не дают способа установления связи с
исходными переменными. Другой подход -- идти от геометрии исходных
переменных и пытаться найти такие проекции фазового пространства, в
которых интегральные многообразия получили бы наиболее простой вид.
Если это удается (а общих способов для этого, к сожалению, не
существует), то связь с исходными переменными сохраняется
изначально. В двух из трех критических подсистем волчка Ковалевской
в двойном поле задача построения алгебраического решения с
разделением переменных решена \cite{KhSavDan,KhND06,KhRCD09}.

Поскольку новые решения приводят к задачам с высокими кратностями
накрытия фазовым пространством плоскости переменных разделения, то
необходимо развитие новых методов не только построения самих
разделений, но и грубого топологического анализа возникающих систем,
в которых простой перебор вариантов оказывается уже нерентабельным.
Мы не затрагиваем здесь более тонкие вопросы вычисления инвариантов
Фоменко~-- Цишанга и чисел вращения, характеризующие системы с
точностью до лиувиллевой или траекторной эквивалентности, хотя при
наличии явных зависимостей фазовых переменных от переменных
разделения (многозначного отображения, которое ниже мы называем
надстройкой) и эта задача решается \cite{BolFomRus}.

Практически для всех классических систем грубый топологический
анализ полностью завершен. Для интегрируемых задач динамики твердого
тела в поле сил с осесимметричным потенциалом это сделано прямыми
аналитическими вычислениями в работах
\cite{Kh76,Kh79,Kh831,Kh832,Kh84}. Полное изложение соответствующих
результатов имеется в \cite{KhBk}, а сами результаты включены,
например, в книгу \cite{BolFomRus}. Тем не менее, на протяжение
почти 20 последних лет регулярно появляются публикации, претендующие
на решение этих же проблем новыми, алгебраическими, методами. Чаще
всего это связано с попытками вывести универсальные способы описания
фазовой топологии исходной системы с помощью исследования свойств
алгебраических кривых, порожденных представлениями Лакса, избежав
при этом выписывания и исследования функций, задающих надстройку.
Очевидно, что полностью игнорировать эти зависимости нельзя.
Действительно, достаточно взять два экземпляра исходной системы,
имеющие одно и то же алгебраическое представление, чтобы получить
удвоенное количество связных компонент в составе интегральных
многообразий. Этот пример тривиален, но можно построить и более
сложные, например, в случае, когда фазовое пространство
неодносвязно. В целом, исследования такого рода важны с
математической точки зрения и способствуют развитию теории. Однако
многие из них страдают рядом принципиальных недостатков. Во-первых,
алгебраические методы применяются лишь к задачам, ранее решенным
иными, более простыми средствами. Ни одной новой интегрируемой
системы фактически исследовано не было. Во-вторых (и это более
важно, поскольку объясняет отсутствие новых результатов
применительно к динамике), на сегодня многочисленным исследователям
этого направления не удалось сформулировать и доказать каких-либо
универсальных теорем, дающих возможность получить описание фазовой
топологии без обращения к отображению надстройки. В своем стремлении
к алгебраизации процесса топологического анализа определенные группы
авторов работ этого направления вообще не приводят каких-либо
достаточно полных доказательных вычислений. Имеются явные ошибки и
заимствования. Отметим, например, статью \cite{Ouz}, которая
практически дословно излагает главу книги \cite{KhBk}, посвященную
исследованию гиростата Чаплыгина~-- Сретенского. Авторы \cite{Ouz}
еще в 2000 году были подробно информированы редакцией журнала
<<Reviews in Mathematical Physics>> о всех имеющихся в этой задаче
опубликованных результатах. Если учитывать, что статья \cite{Ouz}
возникла во французской школе Л.Н.\,Гаврилова, где читаются работы
на русском языке, то подобные публикации вызывают лишь недоумение.
Отметим, что автором из той же научной школы еще раньше, чем
\cite{Ouz}, были выпущены публикации \cite{Rand1,Rand}, в которых
новых результатов в топологии и динамике (по сравнению, например, с
\cite{KhBk}) также не содержится, но в \cite{Ouz} не упомянуты даже
эти работы.

Еще один пример -- книга \cite{Oden}, в которой в качестве авторов
утверждения о характере бифуркаций вдоль всех путей в пространстве
первых интегралов классического случая Ковалевской указаны Харламов,
Оден и Силол. При этом приводятся ссылки на работы \cite{Kh832} и
\cite{OdSil}. Не говоря уже о некоторой очевидной разнице в годах
публикации этих работ, заметим, что в работе \cite{OdSil} все
результаты о перестройках в случае Ковалевской в разделе <<La
methode classique>> взяты из работы \cite{Kh832} с соответствующей
ссылкой. Правда, этот раздел помещен {\it после} описания
бифуркаций, полученных <<новыми>> методами, что создает у читателя
ощущение и новизны самих бифуркаций. Так, например, рисунок
критической поверхности, представляющей собой нетривиальное
расслоение над восьмеркой со слоем окружность, в работе \cite{OdSil}
воспроизведен факсимильно из работы \cite{Kh832}, но как новый,
вдали от ссылок на первоисточник. Кроме того, в той же работе
\cite{OdSil} отмечается, что формула Comessatti \cite{Comes1},
примененная в работе \cite{FranSil} к случаю Ковалевской, привела к
результатам, отличным от результатов Харламова \cite{Kh832}, в связи
с чем необходимы дополнительные вычисления. Впрочем, строгих и
полных вычислений в \cite{Oden,OdSil,FranSil} так и не приводится,
что, как отмечалось, является типичным для упомянутых публикаций.
Отметим связанную с этим циклом чисто реферативную статью
\cite{Les1}, которая в аннотациях подается как совершенно
оригинальная; позже именно этой статье ее автор приписывает и само
происхождение термина <<полная алгебраическая интегрируемость>>
\cite{Les2}, хотя существует много публикаций основоположников
данной теории, в которых встречается этот термин (см., например,
название доклада \cite{VanM} и другие, ставшие уже классическими,
труды Адлера и ван Мёрбеке на эту же тему). Таким образом, несмотря
на высокий интерес в мире к новым результатам в этом направлении,
участие авторов статей \cite{Oden,OdSil,FranSil,Rand,Ouz} и других
подобных публикаций в описании фазовой топологии классических
случаев интегрируемости в динамике твердого тела и ряда их обобщений
весьма незначительно.

В настоящей работе предлагается способ алгоритмизации описания
фазовой топологии алгебраически разрешимых систем с помощью
использования булевых вектор-функций, связанных с накрытием фазовым
пространством пространства вспомогательных переменных, и фактически
описывающих точеч\-но-мно\-жест\-вен\-ную структуру отображения
надстройки над связной компонентой проекции каждого интегрального
многообразии. В предлагаемой здесь первой части работы рассмотрены
классические примеры. В частности, удается получить простые
доказательства утверждений, обоснование которых ранее требовало
технически сложных и весьма громоздких вычислений.

\section{Формализация топологического анализа}\label{sec2}
\subsection{Основные понятия}
Термин <<алгебраически разрешимая система>> не является
общепринятым. Поясним поэтому постановку задачи. Рассмотрим систему
обыкновенных дифференциальных уравнений на подмногообразии
$\mathcal{P}$ вещественного арифметического пространства
\begin{equation}\label{eq2_1}
\frac{d \mbf{x} }{dt} = {\mbf{X}}({\mbf{x}}),\quad \mbf{x} \in
\mathcal{P}
\end{equation}
относительно набора фазовых переменных ${\mbf{x}}$, допускающую
редукцию к системе уравнений вида
\begin{equation}\label{eq2_2}
\frac {ds_i } {d\tau } = \sqrt {\mstrut V_i (s_i; \mbf{f})}.
\end{equation}
Здесь $s_i$ -- вспомогательные переменные или переменные разделения
(вектор вспомогательных переменных обозначим через $\mbf{s}$),
функции $V_i (s;{\mbf{f}})$ -- многочлены от одной переменной $s$ (с
коэффициентами, зависящими от набора произвольных постоянных
${\mbf{f}}$), <<приведенное время>> $\tau$ связано с реальным
временем $t$ зависимостью ${d\tau}/{dt} = T(\mbf{s} ) > 0$. В
основном, такие разделения переменных имеют механическое
происхождение, в связи с чем уравнения (\ref{eq2_2}) имеют структуру
интеграла энергии. При этом предполагается, что все фазовые
переменные $x_j$ выражены через вспомогательные рациональными
функциями от набора радикалов вида
\begin{equation}\label{eq2_3}
R_{i\gamma }  = \sqrt {\mstrut s_i  - e_\gamma  } \quad \quad
(e_\gamma \in {\mathbb{C}})
\end{equation}
с коэффициентами, гладко зависящими от $\mathbf{s}$. В совокупность
чисел $\{ e_\gamma\}$, зависящих, конечно, от постоянных
${\mbf{f}}$, включим и корни многочленов $V_i$. Как показывают
классические примеры, именно корни $V_i $ и исчерпывают обычно все
множество~$\{e_\gamma\}$. Заметим, что частным случаем этой ситуации
являются и допускающие явные алгебраические решения системы с одной
степенью свободы (например, решение для угловых скоростей в случае
Эйлера) и многие точные решения в динамике твердого тела, сводящиеся
к зависящим от произвольного числа параметров семействам
периодических решений и их бифуркаций. В этих случаях
вспомогательная переменная одна, и в этом смысле <<переменные
разделены>>.

Для системы (\ref{eq2_1}) векторный параметр ${\mbf{f}}$ является
набором произвольных постоянных некоторой совокупности ${\mathcal
F}$ первых интегралов на $\mathcal{P}$ (общих или частных).
Зависимости
\begin{equation}\label{eq2_4}
{\mbf{x}} = {\mbf{x}}({\mbf{s}};{\mbf{f}}),\quad {\mbf{s}} \in
\A({\mbf{f}})
\end{equation}
представляют собой (многозначные) параметрические уравнения
интегрального многообразия ${\mathcal F}^{-1}({\mbf{f}})$. Здесь
$\A({\mbf{f}})$ -- область в пространстве вспомогательных
переменных, которая заполняется траекториями системы (\ref{eq2_2})
при заданном~${\mbf{f}}$. Ее традиционно называют не очень
элегантным термином <<область возможности движения>> (ОВД). В
английских переводах предложен более короткий термин -- достижимая
область (accessible region), который вполне отражает суть дела и
которым обусловлено обозначение области аргументов в (\ref{eq2_4}).
Само многозначное отображение (\ref{eq2_4}) будем называть
надстройкой над достижимой областью.

Пусть система (\ref{eq2_1}) есть гамильтонова вполне интегрируемая
система с двумя степенями свободы, имеющая интегральное отображение
$\mathcal{F}:\mathcal{P} \to {\bR}^2$ компактного характера. Две
степени свободы выбраны для простоты изложения, а также потому, что
в рассмотренных ниже примерах число степеней свободы не превосходит
двух.

\begin{defin}\label{thdef21}
Множество $\mathop{\rm Im}\nolimits \mathcal{F}$ назовем допустимой
областью $($в пространстве констант первых интегралов$).$
Соответственно, точка ${\mbf{f}} \in {\bR}^2 $ называется
допустимой, если интегральное многообразие $\mathcal{F}_\mbf{f} =
\mathcal{F}^{ - 1} (\mbf{f})$ не пусто.
\end{defin}

Разделение переменных -- это, вообще говоря, некоторое отображение
\begin{equation}\label{eq2_6}
\pi :\mathcal{P}  \to \bR^2(s_1 ,s_2 ).
\end{equation}
фазового пространства на плоскость вспомогательных переменных $(s_1
,s_2 )$, которое переводит систему (\ref{eq2_1}) в систему уравнений
(\ref{eq2_2}). Ясно, что при фиксированном ${\mbf{f}}$ мы имеем
вполне определенное отображение
\begin{equation}\label{eq2_7}
\pi _{\mbf{f}} :\mathcal{F}_{\mbf{f}}  \to \bR^2(s_1 ,s_2).
\end{equation}

\begin{defin}\label{thdef22}
Для заданного ${\mbf{f}}$  назовем достижимой областью образ
интегрального многообразия на плоскости вспомогательных переменных
$\A({\mbf{f}}) = \pi _{\mbf{f}} (\mathcal{F}_{\mbf{f}})$.
\end{defin}
Ясно, что $\A({\mbf{f}})$ есть подмножество в
\begin{equation}\label{eq2_8}
\{ (s_1 ,s_2 ):V_i (s_i ,{\mbf{f}}) \geqslant 0,i = 1,2\},
\end{equation}
причем это включение таково, что любая точка входит в
$\A({\mbf{f}})$ только вместе со всей своей связной компонентой в
множестве (\ref{eq2_8}). Множество $\A({\mbf{f}})$ есть,
следовательно, совокупность прямоугольников на плоскости $(s_1,s_2)$
(допускаются и <<бесконечные>> прямоугольники -- полуполосы и даже
квадранты). Тот факт, что достижимая область для всех постоянных
интегрирования есть прямоугольник (в указанном обобщенном смысле),
очевидно, является {\it необходимым} условием разделения переменных.
Вероятно, вооружившись определенными предположениями о
нетривиальности проекции вида (\ref{eq2_7}) или о существенной
зависимости уравнений границ достижимой области от параметров
$\mbf{f}$ в терминах неравенства нулю некоторых якобианов, можно
доказать и достаточность, но такая задача здесь не ставится.

Если ${\mbf{f}}$ -- допустимая точка и регулярное значение
интегрального отображения, то $\mF_{\mbf{f}}=\mF^{-1}(\mbf{f})$
состоит из конечного числа двумерных торов Лиувилля с
условно-пе\-ри\-оди\-чес\-ким движением. Задача грубого
топологического анализа практически сводится к следующим этапам:

-- вывод уравнений некоторого множества $\tilde \Sigma  \subset
\bR^2 $  (разделяющего множества), содержащего в себе бифуркационную
диаграмму $\Sigma $ интегрального отображения;

-- определение допустимой области $\mathop{\rm Im}\nolimits
\mathcal{F}$ и, как следствие, описание диаграммы $\Sigma  = \tilde
\Sigma  \cap \mathop{\rm Im}\nolimits \mathcal{F}$;

-- вычисление для каждой связной области в составе $\mathop{\rm
Im}\nolimits \mathcal{F}\, \backslash \Sigma $ количества связных
компонент многообразия $\mathcal{F}_{\mbf{f}}$ (это количество
неизменно, если ${\mbf{f}}$ не пересекает $\Sigma$);

-- определение топологического типа критических интегральных
поверхностей $\mathcal{F}_{\mbf{f}} $, ${\mbf{f}} \in \Sigma $.

Рассмотрим способы реализации этих этапов при наличии
алгебраического разделения переменных.

\subsection{Разделяющее множество и допустимая область} Фиксируем
${\mbf{f}}$, обозначим $v_i  = ds_i /d\tau $, и запишем уравнения
(\ref{eq2_2}) в виде
\begin{equation}\label{eq2_9}
v_i^2  = V_i (s_i ;{\mbf{f}}),\quad i = 1,2.
\end{equation}
Пусть многочлен $V(s;{\mbf{f}})$ -- наименьшее общее кратное
многочленов $V_i $. Назовем $V$ максимальным многочленом системы
(\ref{eq2_2}). <<Хорошие>> разделения переменных обладают тем
свойством, что при условии независимости первых интегралов,
составляющих $\mathcal{F}$, отображение многообразия
$\mathcal{F}_{\mathbf{f}}$ на множество (\ref{eq2_9}) регулярно, а
значит, бифуркациям интегральных многообразий обязательно
соответствует наличие особой точки на множестве (\ref{eq2_9}) в
пространстве $\bR^2(s_1 ,v_1)  \times \bR^2(s_2 ,v_2)$.
Следовательно, в качестве множества $\tilde \Sigma $ выступает
дискриминантная поверхность максимального многочлена. Точки $\tilde
\Sigma $ общего положения (при отсутствии каких-либо особых
вырождений) можно описать параметрическими уравнениями ${\mbf{f}} =
{\mbf{f}}(s)$, полученными из условий
\begin{equation}\label{eq2_10}
V(s;{\mbf{f}}) = 0,\quad V'_s (s;{\mbf{f}}) = 0.
\end{equation}

Если все $V_i$ одинаковы и равны $V$, то, считая $s,v \in
\mathbb{C}$, получим уравнение алгебраической кривой
$v^2=V(s;{\mbf{f}})$. Тогда (\ref{eq2_10}) -- условия наличия у этой
кривой особой точки. В системах более общего вида, связанных с
алгебраическими кривыми, например, через представление Лакса, особые
участки разделяющего множества (и, соответственно, бифуркационных
диаграмм) возникают при условиях приводимости алгебраических кривых.
В диаграммах на плоскости условия приводимости порождают, в
частности, изолированные точки (случай Клебша \cite{PogKh}, Лагранжа
\cite{Oshem1}, волчок Ковалевской в двойном поле \cite{KhRCD1}).

Определение допустимой области непосредственно связано с процессом
выявления условий вещественности зависимостей (\ref{eq2_4}),
описывающих отображение, обратное к (\ref{eq2_6}).

Пусть $\{ e_\gamma \}$ -- множество корней максимального многочлена.
Считаем пока, что все они вещественны. Что делать при наличии пар
комплексно сопряженных корней, скажем позже. Пусть система является
алгебраически разрешимой. Соответствующие радикалы (\ref{eq2_3})
назовем базисными. Многозначное отображение (\ref{eq2_4}) записано
по предположению рациональными функциями от базисных радикалов с
коэффициентами, гладко зависящими от $s_1 ,s_2 ,{\mbf{f}}$. Задача
сводится к определению тех прямоугольников (компонент) в составе
множества (\ref{eq2_8}), для которых значения надстройки
${\mbf{x}}(s_1 ,s_2 ;{\mbf{f}})$ являются вещественными. Без
ограничения общности можно рассматривать лишь случай отсутствия
кратных корней у максимального многочлена, так как при этом мы
определим внутренность допустимой области, а всю область получим,
переходя к замыканию. Для проверки условий вещественности, конечно,
достаточно взять точку прямоугольника, не являющуюся граничной, то
есть точку, в которой все~$V_i$ отличны от нуля.

Любую рациональную функцию от радикалов (\ref{eq2_3}) можно
преобразовать в многочлен от этих же радикалов с коэффициентами,
гладко зависящими от $s_1$, $s_2$ и ${\mbf{f}}$ (домножая числитель
и знаменатель на подходящие многочлены). Выделим множество мономов
$Z_1 ,...,Z_k$ (произведений радикалов (\ref{eq2_3}) в различных
сочетаниях), от которых исходные фазовые переменные ${\mbf{x}}$
зависят линейно
\begin{equation}\label{eq2_11}
{\mbf{x}} = {\mbf{a}}+ \sum {\mbf{b}_j} Z_j
\end{equation}
(${\mbf{a}}, {\mbf{b}_j}$ -- однозначные вектор-функции от
${\mbf{s}}$ и параметров ${\mbf{f}}$). Тогда условия вещественности
значений отображения (\ref{eq2_4}) можно записать в виде системы
неравенств
\begin{equation}\label{eq2_12}
\varepsilon _j Z_j^2  \geqslant 0,\quad \varepsilon _j^2  = 1\quad
(j = 1,...,k).
\end{equation}

Пусть $\mB = \{ 0,1\} $ -- булева пара.
\begin{defin}\label{thdefbs}
Булевым знаком назовем функцию $\lsgn : \bR \to \mB$, такую, что
\begin{equation}\notag
\lsgn (\theta ) = \left\{
\begin{array}{l}
{0,\quad \theta  \geqslant 0}  \\ {1,\quad \theta  < 0}
\end{array}
\right. .
\end{equation}
\end{defin}

Обозначая символом $\oplus$ сумму по модулю $2$, имеем свойство
\begin{equation}\label{eq2_13}
\lsgn (\theta _1 \theta _2 ) = \lsgn (\theta _1 ) \oplus \lsgn
(\theta _2 ).
\end{equation}

Пусть $U_1 ,...,U_N $  -- совокупность всех радикалов вида
(\ref{eq2_3}), участвующих в записи отображения (\ref{eq2_4}).
Соответствие знаков величин $Z_j^2 $  знакам величин $U_i^2 $
запишется некоторой булевой вектор-функцией с аргументами $u_i  =
\lsgn (U_i^2 )$ и значениями $z_j  = \lsgn (Z_j^2 )$, вычисляемой с
помощью операций $ \oplus $. Условия вещественности (\ref{eq2_12})
записываются в виде заданного единственного значения $\mbf{z}_0$
этой вектор-функции. Определив точки $\mbf{u}$ в ее прообразе,
получим соответствующие системы неравенств, содержащих переменные
$s_i$, которые необходимо затем разрешать относительно $s_i$ при
различных наборах параметров $\mbf{f}$.

Далее для термина <<булева вектор-функция>> будем иногда
использовать сокращение БВФ. Заметим, что термин <<значение>> в
применении к функциям всегда несет в себе двусмысленность (значение
как функция аргумента и конкретное значение-константа в точке).
Поэтому для булевых вектор-функций скалярную булеву функцию в
составе вектора значений будем называть компонентой. Термин <<булева
функция>> предполагает одну компоненту (скаляр).

Процесс решения полученных систем неравенств может быть упрощен,
если заранее известно, как расположены числа $e_\gamma$.
Действительно, если ${e_\gamma   < e_\delta}$, то из неравенства
${s_i - e_\delta>0}$ заведомо следует, что ${s_i - e_\gamma>0}$, и
соответствующие булевы аргументы (пусть это будут $u_\delta^{(i)}
,u_\gamma^{(i)}$) не произвольны, а связаны некоторым условием. Это
условие выглядит так: если $u_\delta^{(i)}  = 0$, то $u_\gamma^{(i)}
= 0$, или в терминах булевых функций
\begin{equation}\label{eq2_14}
(u_\gamma^{(i)}  \to u_\delta^{(i)} ) = 1
\end{equation}
(стрелкой обозначена, как обычно, импликация). Упорядочив все корни
$e_\gamma$ по возрастанию, добавим в компоненты вводимой БВФ для
каждой пары соседних корней $e_\gamma   < e_\delta$ булевы функции
вида ${u_\gamma^{(i)}  \to u_\delta^{(i)}}$, по одной для каждой
переменной $s_i$. Кроме того, часто бывает, что найденные выражения
для отображения \eqref{eq2_11} симметричны по переменным $s_1,s_2$,
или же, что из самого определения этих переменных вытекает
возможность условиться о неравенстве типа $s_2 < s_1$. Тогда для
любого корня $e_\gamma$ имеется свойство: из неравенства $s_2 -
e_\gamma>0$ заведомо следует, что $s_1 - e_\gamma>0$. Поэтому
соответствующие булевы аргументы, для примера
$u_{\gamma}^{(1)},u_{\gamma }^{(2)}$, связаны условием
\begin{equation}\label{eq2_15}
(u_{\gamma}^{(1)}  \to u_{\gamma}^{(2)} ) = 1,
\end{equation}
так что для всех $e_\gamma$ можно в компоненты вводимой БВФ добавить
функции вида ${u_{\gamma}^{(1)} \to u_{\gamma}^{(2)}}$. После этого
потребуем, чтобы в соответствии с (\ref{eq2_14}), (\ref{eq2_15}) все
компоненты-импликации равнялись 1.

Подчеркнем, что добавление компонент-импликаций лишь упрощает
нахождение областей $\A(\mbf{f})$, поскольку для того, чтобы
выписать все неравенства на предыдущем этапе, необходимо в любом
случае учесть взаимное расположение значений в множестве
$\{e_\gamma\}$. Добавление компонент-импликаций, если оно возможно,
позволяет достичь и большей формализации всех шагов.

Обозначим полученную в результате булеву вектор-функцию через
\begin{equation}\notag
{\mbf{z}} = {A}({\mbf{u}}),\quad {\mbf{u}} \in \mB^N ,\quad
{\mbf{z}} \in \mB^{k'} \qquad (k' \geqslant k).
\end{equation}
Условие существования вещественных значений отображения
(\ref{eq2_4}) при заданном наборе констант~${\mbf{f}}$ запишется в
виде условия на промежутки изменения $s_i $, определяющие достижимую
область $\A({\mbf{f}})$ в следующей простой форме ${\mbf{u}} \in
{A}^{ - 1} ({\mbf{z}}_0 )$, где ${\mbf{z}}_0  \in \mB^{k'}$ вполне
определено условиями (\ref{eq2_12}), (\ref{eq2_14}), (\ref{eq2_15}).

\subsection{Вычисление количества связных компонент} Фиксируем
параметры ${\mbf{f}}$, прямоугольник $\Pi$ в составе $\A({\mbf{f}})$
и поставим вопрос -- каково количество компонент связности в его
прообразе при отображении (\ref{eq2_4})? Пусть в
прямоугольнике~$\Pi$ (для примера, конечном)
\begin{equation}\label{eq2_16}
s_1  \in [e_\gamma, e_\delta],\quad s_2  \in [e_\lambda, e_\mu].
\end{equation}
Тогда вдоль любой траектории системы (\ref{eq2_2}), заключенной в
этих границах, периодически меняют знак радикалы
\begin{equation}\label{eq2_17}
R_{1\gamma } ,R_{1\delta } ,R_{2\lambda } ,R_{2\mu },
\end{equation}
а вместе с ними и все значения мономов $Z_j$, в которые входят эти
радикалы. Таким образом, все точки ${\mbf{x}}$, отличающиеся только
знаками указанных $Z_j$, будут принадлежать одной связной компоненте
$\mathcal{F}_{\mbf{f}}$. Фиксируем точку ${\mbf{s}} = (s_1 ,s_2 )$,
внутреннюю для $\Pi$. Мономы $Z_j$ следует выбирать <<разумно>>, то
есть так, чтобы зависимости \eqref{eq2_11} формально порождали
взаимно однозначное соответствие множества точек $\mbf{x}$,
накрывающих точку $\mbf{s}$, с множеством $\mB^k$ наборов знаков
величин $Z_j$ хотя бы для почти всех $\mbf{s} \in \Pi$. Отметим, что
$Z_j$ -- произведения радикалов, поэтому они могут быть, как и сами
радикалы (\ref{eq2_3}), вещественными или чисто мнимыми. Однако на
одной связной компоненте интегрального многообразия и, более того,
на одном выбранном прямоугольнике~$\Pi$ подкоренные выражения не
могут менять знак. Поэтому всегда можно <<подправить>> определения
базисных радикалов, меняя при необходимости знак подкоренного
выражения и домножая на нужное количество мнимых единиц
соответствующие коэффициенты в представлении \eqref{eq2_11}, так,
чтобы и радикалы, и коэффициенты на множестве \eqref{eq2_16} были
вещественными.

Определим теперь булеву вектор-функцию ${C}({\mbf{u}})$, сопоставив
набору булевых переменных $u_i  = \lsgn (U_i)$ набор булевых
переменных $z_j = \lsgn (Z_j)$. Отметим, что $Z_j$ выражаются через
$U_i$ по тем же правилам, по каким $Z_j^2$ выражаются через $U_i^2$.
Следовательно, если на этапе определения допустимого множества и
достижимых областей мы не использовали переменных, отличных от
${\mbf{x}}$, то новая БВФ окажется проекцией использованной на
предыдущем этапе булевой вектор-функции на первые координаты,
определенные операцией сложения по модулю~2 (игнорируются
компоненты, заданные импликацией). Кроме того, как следует из
свойства (\ref{eq2_13}), все компоненты -- линейные функции
аргументов $u_i$.

Выбрав некоторую область (\ref{eq2_16}), разобьем аргументы
${\mbf{u}}$ на две группы. В первую группу войдут булевы знаки
${\mbf{v}}\in \mB^m$ радикалов, не изменяющих знака вдоль
траектории, а во вторую группу ${\mbf{w}}\in \mB^n$ включим булевы
знаки радикалов, меняющих знак вдоль траектории периодически. В
данном примере ко второй группе относятся аргументы, отвечающие
радикалам (\ref{eq2_17}), в первую -- все остальные. Определенная
выше БВФ примет вид
\begin{equation}\label{eq2_18}
{C}: \mB^m \times \mB^n \to \mB^k.
\end{equation}

\begin{defin}\label{thdef1} Пусть задана булева вектор-функция
${B}: \mB^m \times \mB^n \to \mB^k$. Будем говорить, что элементы
$\mbf{v}',\mbf{v}'' \in \mB^m$ эквиваленты относительно ${B}$, если
существуют $\mbf{w}',\mbf{w}'' \in \mB^n$ такие, что
${B}({\mbf{v}'},{\mbf{w}'}) = {B}({\mbf{v}''},{\mbf{w}''})$.
\end{defin}

Эквивалентность относительно ${B}$ обозначим $ \mbf{v}' \msim{{B}}
\mbf{v}''$. Аргументы, входящие в состав вектора $\mbf{v}$ (классы
эквивалентности которого определяются), назовем аргументами первой
группы. Остальные будем называть аргументами второй группы.

Из построения отображения (\ref{eq2_18}) вытекает следующее
утверждение.
\begin{theorem}
Количество $c({\mbf{f}},\Pi)$ связных компонент многообразия
$\mathcal{F}_{\mbf{f}} $, накрывающих прямоугольник $\Pi$, равно
количеству классов эквивалентности в множестве $\mB^m$ по отношению
эквивалентности относительно функции ${C}$.
\end{theorem}

Очевидно, класс эквивалентности вектора ${\mbf{v}} \in \mB^m$
относительно ${C}$ определяется равенством
\begin{equation}\notag
\mC_{{C}}({\mbf{v}})  = p_1({C}^{-1}
(\{\mbf{z}={C}(\mbf{v},\mbf{w}): \mbf{w} \in \mB^n\})), \qquad
p_1(\mbf{v},\mbf{w}) \equiv \mbf{v}.
\end{equation}
Если все классы имеют одно и то же количество элементов
$\mathcal{N}({\mbf{f}},\Pi)$, то
\begin{equation}\notag
c({\mbf{f}},\Pi ) = 2^{m}/\mathcal{N}({\mbf{f}},\Pi).
\end{equation}

Нетрудно описать теперь сам алгоритм вычисления этого значения.
Составим таблицу булевой вектор-функции
${\mbf{z}={C}(\mbf{v},\mbf{w})}$, содержащую $2^N$ строк
($N=\dim(\mbf{v},\mbf{w})$). Назовем строку помеченной числом $i$,
если уже установлена ее принадлежность классу эквивалентности с
порядковым номером $i$. Пусть $i_0$ классов строк уже помечены.
Вводим пустые массивы $\mbf{K}$ и $\mbf{Z}$. Находим первую
непомеченную строку (если таких нет -- процесс завершен).
Присваиваем ей метку $i_1=i_0+1$, соответствующий вектор $\mbf{v}$
включаем в новый класс с номером ${i_1}$ (массив $\mbf{K}$), а
соответствующее значение $\mbf{z}$ -- в массив $\mbf{Z}$.
Просматриваем строки таблицы, начиная с текущей. Если строка
непомечена и у нее $\mbf{v}\in \mbf{K}$ или $\mbf{z}\in \mbf{Z}$, то
помечаем ее номером $i_1$, и если одно из включений не выполнено, то
пополняем соответствующий массив. При этом просмотр таблицы
повторяется до тех пор, пока происходит пополнение хотя бы одного из
массивов $\mbf{K}$ или $\mbf{Z}$ новыми элементами. Класс с номером
$i_1$ сформирован. Полагаем $i_0=i_1$ и повторяем процесс с ввода
пустых массивов.

\begin{remark}\label{rem1}
Мы рассмотрели случай, когда все корни в множестве $\{ e_\gamma\}$
вещественны. Что будет, если среди корней имеется пара комплексно
сопряженных $e_\delta   = \overline {e_\gamma} $? Тогда в
вещественных выражениях для переменных ${\mbf{x}}$ соответствующая
пара радикалов обязательно будет входить либо в виде произведения,
либо в некоторой комбинации вида
\begin{equation}\label{eq2_19}
c_\gamma  R_{i\gamma }  + c_\delta  R_{i\delta } ,\quad c_\gamma   =
\overline {c_\delta}.
\end{equation}
В первом случае в качестве базисных нужно взять радикалы
$$
\sqrt{s_i^2 - (e_\gamma+\overline {e_\gamma})s_i+e_\gamma \overline
{e_\gamma}}
$$
и считать их не меняющими знак. Во втором случае радикалы
$R_{i\gamma} ,R_{i\delta}$ не меняют знак вдоль траектории, но они
изначально комплексно сопряжены, поэтому выбор их знаков в начальный
момент времени на траектории не произволен, а подчинен условию
$R_{i\gamma} R_{i\delta} > 0$. Чтобы ему удовлетворить, достаточно
ввести в булеву вектор-функцию еще по одной компоненте вида
$u_\gamma^{(i)} \oplus u_\delta^{(i)}$ для каждого $s_i$ и
потребовать, чтобы они равнялись булевой константе~$0$. В алгоритме
прямого вычисления количества классов эквивалентности это, конечно,
равносильно отбрасыванию определенной части значений аргумента
${\mbf{u}}$, то есть установке фильтра в таблице, реализующей
функцию ${C}(\mbf{u})$. Произвол в выборе знака при этом сохранится
у выражения~$(\ref{eq2_19})$.
\end{remark}

Ниже мы развиваем определенную технику работы с БВФ для вычисления
количества классов эквивалентности. Она полностью автоматизирует
процесс исследования регулярных случаев.

Для критических интегральных поверхностей известные на данный момент
алгебраически разрешимые системы поддаются формальному анализу теми
же методами с несложными дополнительными рассуждениями. На сегодня
все основные типы бифуркаций в системах с двумя степенями свободы
известны. По сравнению с теми перестройками, которые были выявлены в
классических задачах динамики твердого тела \cite{KhBk}, никаких
новых не открыто (за исключением неориентируемого случая). Несмотря
на очевидную возможность возникновения разнообразных комбинаций и
склеек, фактически в реальной задаче встретилась лишь одна новая
перестройка такого рода \cite{OrRyab}. Полная теоретическая
классификация невырожденных особенностей с возникающими атомами
получена в \cite{BolFomRus}. В большинстве случаев при наличии
бифуркации $\mathcal{F}_{{\mbf{f}}_ -  }  \to
\mathcal{F}_{{\mbf{f}}_0 } \to \mathcal{F}_{{\mbf{f}}_ +  }$ знание
количества связных компонент всех трех указанных поверхностей
позволит установить и топологический тип критической поверхности
$\mathcal{F}_{{\mbf{f}}_0}$. Как отмечено выше, бифуркация
сопровождается возникновением кратного корня у максимального
многочлена, то есть, по крайней мере, у одного из многочленов $V_i$.
Как показывает опыт исследования алгебраически разрешимых систем,
вариантов может быть очень много. Однако в рамках рассматриваемой
здесь задачи принципиально важны только два: один из отрезков
(\ref{eq2_16}) стягивается в точку или кратный корень возникает
внутри (в том числе, возможно, что и в граничной точке) этого
отрезка. В первом случае соответствующий радикал тождественно равен
нулю и <<раздвоения>> в прообразе вызывать не может. Его следует,
поэтому, отнести к группе радикалов, меняющих знак. Во втором случае
заметим, что точка, отвечающая кратному корню, с точки зрения
исследования количества связных компонент, <<проходится>>
траекторией в бесконечные моменты времени $t =  \pm \infty$, и при
этом возникшие на месте пары сопряженных радикалов вещественные
выражения следует также считать меняющими знак вдоль условной
траектории в плоскости $(s_i ,v_i)$, геометрически совпадающей с
соответствующей <<восьмеркой>>. Интересные явления возникают, когда
кратный корень лежит за пределами достижимой области. Он может
вообще не повлиять на интегральную поверхность, накрывающую
рассматриваемый прямоугольник, и тогда лиувиллев тор непрерывно
меняется без бифуркаций, но может оказать и некое <<внешнее>>
воздействие, которое повлечет за собой, например, стягивание одной
из образующих в точку. Лишь в таких случаях может потребоваться
дополнительный анализ формул.

\section{Редукция булевых вектор-функций}\label{sec3}
Приведем ряд утверждений, позволяющих сократить вычисления при
оценке количества классов эквивалентности, а в некоторых случаях
сделать результат совершенно наглядным.

Напомним, что аргумент $u_j$ булевой функции называется фиктивным,
если ее значение не меняется при инверсии этого аргумента (замене
$u_j$ на $\neg u_j$). Компоненту булевой вектор-функции назовем
фиктивной, если она равна константе.

\begin{lemma}\label{lem1}
Количество классов эквивалентности относительно булевой
вектор-функции не изменяется при отбрасывании фиктивных аргументов
любой группы или фиктивных компонент.
\end{lemma}
Доказательство очевидно. Заметим, что исключение фиктивного
аргумента второй группы и компоненты-константы не изменяет и самих
классов эквивалентности. Далее предполагаем, что в исходных функциях
фиктивные аргументы и компоненты исключены.

\begin{defin}\label{thdef2}
Пусть $u_1,...,u_N$ булевы переменные. Булевым мономом от этих
переменных назовем выражение вида
\begin{equation}\notag
u_{j_1}\oplus ... \oplus u_{j_M}, \qquad M \leqslant N.
\end{equation}
\end{defin}
Поскольку $u \oplus u \equiv 0$, $u \oplus 0 \equiv u$, то любую
переменную имеет смысл включать в моном лишь в первой степени.

\begin{defin}\label{thdef3}
Алгебраической булевой вектор-функцией $($или, сокращенно, АБВФ$)$
назовем такую булеву вектор-функцию, все компоненты которой являются
булевыми мономами от аргументов.
\end{defin}

Очевидно, алгебраичность БВФ означает, что она является линейным
отображением векторных пространств над полем $\mB$. Далее в разделе
предполагается, что задана АБВФ вида
\begin{equation}\label{eq3_1}
{C}: \mB^m \times \mB^n \to \mB^k,
\end{equation}
где $\mB^m=\{\mbf{v}\}$ -- пространство аргументов первой группы,
$\mB^n=\{\mbf{w}\}$ -- пространство аргументов второй группы.

\begin{lemma}\label{lem2}
Пусть функция
$$
{B}: \mB^m \times \mB^n \to \mB^{k_1},  \qquad {k_1}<k
$$
получена из функции ${C}$ отбрасыванием нескольких компонент
$($проекцией на пространство меньшей размерности$).$ Тогда из
$\mbf{v}'\msim{C}\mbf{v}''$ следует, что
$\mbf{v}'\msim{B}\mbf{v}''$.
\end{lemma}

\begin{lemma}\label{lem3}
Пусть функция ${C}$ имеет в составе своих компонент некоторое
выражение $z_j$, которое можно представить алгебраическим мономом от
других компонент этой функции. Множество классов эквивалентности
относительно ${C}$ в пространстве любой группы аргументов не
изменится, если отбросить компоненту $z_j$.
\end{lemma}

\begin{lemma}\label{lem8} Пусть
заданы $\mB$-линейные изоморфизмы ${X}: \mB^m  \to \mB^m$,
${Y}:\mB^n \to \mB^n$, ${Z}:\mB^k  \to \mB^k$. Рассмотрим функцию
${B}:\mB^m \times \mB^n  \to \mB^k$, определенную как ${B} = {Z}
\circ {C} \circ ({X},{Y})$. Тогда количество классов эквивалентности
в пространстве $\mB^m$ относительно функций ${C}$ и ${B}$ совпадает.
\end{lemma}

Доказательство этих трех лемм вытекает прямо из определения
эквивалентности.

\begin{lemma}\label{lemxn}
Пусть у функции ${C}$ существует набор компонент
$\widetilde{\mbf{z}}\in \mB^{k_1}$, зависящих только от аргументов
$\widetilde{\mbf{w}}\in \mB^{n_1}$ второй группы, причем эти
аргументы, в свою очередь, входят только в эти компоненты. Тогда
классы эквивалентности в пространстве $\mB^m$ относительно функции
${C}$ не изменятся, если отбросить указанные компоненты и аргументы.
\end{lemma}
\begin{proof} При необходимости перенумеруем компоненты и аргументы
так, чтобы компоненты $\widetilde{\mbf{z}}$ и аргументы
$\widetilde{\mbf{w}}$ оказались последними в записи:
${\mbf{z}}=(\widehat{\mbf{z}},\widetilde{\mbf{z}})$, $
{\mbf{w}}=(\widehat{\mbf{w}},\widetilde{\mbf{w}})$. Функция $C$
примет вид
$$
(\widehat{\mbf{z}},\widetilde{\mbf{z}})=C(\mbf{v},\widehat{\mbf{w}},\widetilde{\mbf{w}})
= (B(\mbf{v},\widehat{\mbf{w}}), D(\widetilde{\mbf{w}}))
$$
где $B:\mB^m{\times}\mB^{n-n_1} \to \mB^{k-k_1}$, $D:\mB^{n_1}\to
\mB^{k_1}$. По лемме~\ref{lem2} из $\mbf{v}'\msim{C}\mbf{v}''$
следует, что $\mbf{v}'\msim{B}\mbf{v}''$. Пусть
$\mbf{v}'\msim{B}\mbf{v}''$. Положив $\widetilde{\mbf{w}}=0$,
получим $\widetilde{\mbf{z}}=0$, откуда $\mbf{v}'\msim{C}\mbf{v}''$.
\end{proof}

\begin{lemma}\label{lem4}
Пусть у функции ${C}$ существует аргумент второй группы $w_i$,
входящий ровно в одну компоненту $z_j$. Тогда классы эквивалентности
в пространстве $\mB^m$ относительно функции ${C}$ не изменятся, если
отбросить компоненту $z_j$ и аргумент $w_i$.
\end{lemma}
\begin{proof} Выделим два предположения: во-первых, $w_i$ не
является фиктивным аргументом для~$z_j$; во-вторых, для всех
остальных компонент функции ${C}$ аргумент $w_i$ является фиктивным.
Пусть для определенности $i=n, j=k$, то есть выделенный аргумент и
выделенная компонента являются последними в записи (это можно
сделать перестановками аргументов и компонент). Обозначим
$\widehat{\mbf{w}}=(w_1,...,w_{n-1})$. По условиям леммы можно
записать
\begin{equation}\notag
{C}(\mbf{v},\widehat{\mbf{w}},w_n)=(B(\mbf{v},\widehat{\mbf{w}}),\xi(\mbf{v},
\widehat{\mbf{w}})\oplus w_n),
\end{equation}
где $\xi(\mbf{v}, \widehat{\mbf{w}})$ -- некоторый булев моном. По
лемме \ref{lem2} достаточно доказать, что из
$\mbf{v}'\msim{{B}}\mbf{v}''$ следует $\mbf{v}'\msim{C}\mbf{v}''$.
Пусть существуют $\widehat{\mbf{w}}'$, $\widehat{\mbf{w}}''$ такие,
что
${B}(\mbf{v}',\widehat{\mbf{w}}')={B}(\mbf{v}'',\widehat{\mbf{w}}'')$.
Если
$\xi(\mbf{v}',\widehat{\mbf{w}}')=\xi(\mbf{v}'',\widehat{\mbf{w}}'')$,
то, очевидно, $\mbf{v}'\msim{C}\mbf{v}''$. В противном случае
положим $\mbf{w}'=(\widehat{\mbf{w}}',0)$,
$\mbf{w}''=(\widehat{\mbf{w}}'',1)$. Получим
${C}(\mbf{v}',\mbf{w}')={C}(\mbf{v}'',\mbf{w}'')$. Итак, последнюю
компоненту ${C}$ можно отбросить без изменения классов
эквивалентности. Но тогда аргумент $w_n$ станет фиктивным и его
также можно исключить. Лемма доказана.
\end{proof}

Пусть известно, что некоторый аргумент заведомо принадлежит ко
второй группе, но входит в несколько компонент функции. Можно ли
такую функцию всегда привести к виду, допускающему применение
леммы~\ref{lem4}? Ответ дают следующие две леммы.

\begin{lemma}\label{lem5}
Если к некоторой компоненте функции ${C}$ прибавить $($по
модулю~$2)$ некоторую другую компоненту, то классы эквивалентности
относительно этой функции не изменятся.
\end{lemma}
Для доказательства достаточно в условиях леммы~\ref{lem8} положить
\begin{equation}\label{eq3_2}
\begin{array}{c}
X = {\mathop{\rm Id}\nolimits}_{\mB^m }, \quad Y = {\mathop{\rm Id}\nolimits}_{\mB^n },\\
Z(z_1 , \ldots ,z_i , \ldots ,z_j , \ldots ,z_k ) = (z_1 , \ldots
,z_i  \oplus z_j , \ldots ,z_j , \ldots ,z_k ).
\end{array}
\end{equation}

\begin{remark}\label{remm1}
Перестановки компонент также являются $\mB$-изоморфизмами $\mB^k$.
Аналогично, перестановки аргументов одной группы и подстановки вида
\begin{equation}\label{eq3_3}
u_i  \to u_i  \oplus u_j
\end{equation}
в пределах одной группы служат $\mB$-изоморфизмами пространств
$\mB^m$, $\mB^n$. Поэтому они не меняют количества классов
эквивалентности.
\end{remark}

\begin{lemma}\label{lemcor5}
Пусть у функции ${C}$ существует компонента $z_j=w_i$, где $w_i$ --
аргумент второй группы. Тогда количество классов эквивалентности в
пространстве $\mB^m$ относительно функции ${C}$ не изменится, если
исключить аргумент $w_i$ из всех компонент и затем отбросить
компоненту $z_j$ и аргумент $w_i$.
\end{lemma}

\begin{proof}
Прибавим компоненту $z_j$ ко всем другим компонентам, содержащим
аргумент $w_i$. По лемме~\ref{lem5} классы эквивалентности не
изменятся. При этом аргумент $w_i$ будет исключен из всех компонент,
кроме $z_j$. Но тогда по лемме~\ref{lem4} можно отбросить компоненту
$z_j$ и аргумент $w_i$.
\end{proof}

\begin{remark}\label{rem2}
Если в результате действий, осуществляемых в соответствии с
приведенными утверждениями, окажется, что в функции не осталось
аргументов первой группы или все компоненты подлежат исключению, то
класс эквивалентности, очевидно, один.
\end{remark}

\begin{lemma}\label{lem7}
Пусть функция ${C}$ представима в виде ${C}={B}\circ ({X},{Y})$, где
$$
{{X}: \mB^m  \to \mB^{m_1}}, \qquad {{Y}:\mB^m\times \mB^n \to
\mB^{n_1}}, \qquad {{B}: \mB^{m_1}\times \mB^{n_1} \to \mB^k}
$$
-- алгебраические булевы вектор-функции. Предположим, что

$(i)$ отображение ${X}$ сюръективно;

$(ii)$ для любого $\mbf{v}\in \mB^m$ сюръективно индуцированное
отображение $ {Y}(\mbf{v}, {\boldsymbol \cdot}) : \mB^n \to
\mB^{n_1}$.

Тогда количество классов эквивалентности в $\mB^m$ относительно
${C}$ равно количеству классов эквивалентности в $\mB^{m_1}$
относительно ${B}$.
\end{lemma}

\begin{proof} Имеем
$$
C(\mbf{v},\mbf{w})=B(X(\mbf{v}),Y(\mbf{v},\mbf{w})).
$$
Пусть $\mbf{v}' \msim{{C}} \mbf{v}''$. Тогда, очевидно,
${X}(\mbf{v}')\msim{{B}}{X}(\mbf{v}'')$. Поэтому корректно
определено отображение классов $\mC_{{C}}(\mbf{v}) \mapsto
\mC_{{B}}({X}(\mbf{v}))$. По условию $(i)$ это отображение
сюръективно. Покажем, что оно и взаимно-однозначно. Пусть элементы
${\mbf{y}',\mbf{y}'' \in \mB^{n_1}}$ эквивалентны относительно
${B}$. Тогда для некоторых ${\mbf{z}', \mbf{z}''\in \mB^{m_1}}$
будем иметь ${{B}(\mbf{y}',\mbf{z}')={B}(\mbf{y}'',\mbf{z}'')}$. По
предположению $(i)$ существуют $\mbf{v}',\mbf{v}'' \in \mB^m$, такие
что ${X}(\mbf{v}')=\mbf{y}'$ и ${X}(\mbf{v}'')=\mbf{y}''$. Тогда по
предположению $(ii)$ найдутся $\mbf{w}',\mbf{w}'' \in \mB^n$, такие
что ${Y}(\mbf{v}',\mbf{w}')=\mbf{z}'$ и
${Y}(\mbf{v}'',\mbf{w}'')=\mbf{z}''$. Но это означает, что
${C}(\mbf{v}',\mbf{w}')={C}(\mbf{v}'',\mbf{w}'')$, то есть $\mbf{v}'
\msim{{C}} \mbf{v}''$. Таким образом, различным классам
эквивалентности относительно ${C}$ не может соответствовать один
класс эквивалентности относительно ${B}$, ч.т.д.
\end{proof}

\begin{remark}\label{rem3}
Эта лемма позволяет, при определенных условиях, заменять в
отображениях надстройки набор базисных радикалов на составленные из
них мономы и в качестве аргументов АБВФ выбирать булевы знаки
последних.
\end{remark}

Любой булев моном от переменных $u_1 , \ldots ,u_N$ можно
представить в виде
$$
\mathop  \oplus \limits_{i = 1}^N b_i u_i, \qquad b_i  \in \mB
$$
(сумма по модулю $2$, произведение -- стандартное в $\mB$). Поэтому
любая алгебраическая булева вектор-функция $B: \mB^N  \to \mB^k $
записывается с помощью двоичной матрицы
\begin{equation}\notag
|| {b_{ij} } ||: \qquad z_{i}  = \mathop  \oplus \limits_{j = 1}^N
b_{ij} u_j ,\quad i = 1 \div k.
\end{equation}
Не будем различать обозначения АБВФ и ее матрицы.

\begin{defin}\label{thdef5}
Для функции вида $(\ref{eq3_1})$ следующие преобразования  ее
матрицы назовем элементарными:

$(i)$ перестановка строк;

$(ii)$ перестановка столбцов в пределах одной группы;

$(iii)$ прибавление $($по модулю $2)$ к некоторой строке другой
строки;

$(iv)$ прибавление $($по модулю $2)$ к некоторому столбцу другого
столбца той же группы.
\end{defin}

\begin{defin}\label{thdef6}
Пусть дана функция вида $(\ref{eq3_1})$. Назовем эквивалентными
преобразованиями ее матрицы преобразования, полученные умножением на
матрицы слева или справа над полем~$\mB$, которые не меняют
количества классов эквивалентности первой группы аргументов.
\end{defin}

\begin{lemma}\label{lem9}
Все элементарные преобразования $(i)$--$(iv)$ являются
эквивалентными. В дополнение к ним эквивалентными являются следующие
преобразования:

$(v)$ отбрасывание нулевого столбца;

$(vi)$ отбрасывание нулевой строки;

$(vii)$ отбрасывание набора строк и набора столбцов второй группы, в
которых отличные от нуля элементы лежат только в их пересечении;

$(viii)$ отбрасывание строки с номером $i$ и столбца второй группы с
номером $j$, если этот столбец единичный с единицей в строке $i$.
\end{lemma}
\begin{proof} В части элементарных преобразований утверждение следует из
леммы~\ref{lem8} при введении $\mB$-изоморфизмов, порожденных
перестановками и преобразованиями вида (\ref{eq3_2}), (\ref{eq3_3}).
Преобразования $(v)$, $(vi)$ эквивалентны в силу леммы~\ref{lem1},
преобразование $(vii)$ -- в силу леммы~\ref{lemxn}, а преобразование
$(viii)$ -- в силу леммы~\ref{lem4}.
\end{proof}

\begin{lemma}\label{lem11} Пусть $A$ -- двоичная $k{\times}N$-матрица ранга $P$ и все ее
столбцы отнесены к одной группе. Тогда элементарными
преобразованиями она приводится к виду
\begin{equation}\label{eq3_4}
A = \begin{array}{||c|c||}
   E_P & 0  \\
   \hhline{||-|-||}
   0 & 0  \\
 \end{array}\,.
\end{equation}
\end{lemma}
\begin{proof}
Перестановками строк и столбцов добьемся того, чтобы невырожденная
матрица размерности $P{\times}P$ была левым верхним углом $A$.
Приведем ее к единичной методом Гаусса над полем $\mB$. Пусть
${\mbf{a}}_i $ -- строка с номером $i$, ${\mbf{b}}_i$ -- столбец с
номером $i$ ($i \leqslant P$). Если в строке с номером $j > P$ на
месте $i$ стоит 1, то прибавим к ней строку ${\mbf{a}}_i $. Если в
столбце с номером $j
> P$ на месте $i$ стоит 1, то прибавим к нему столбец ${\mbf{b}}_i$.
Получим нулевые блоки в (\ref{eq3_4}) на побочной диагонали.
Оставшийся юго-восточный блок окажется нулевым по предположению о
ранге.
\end{proof}
\begin{defin}\label{thdef7}
Матрицу вида $(\ref{eq3_4})$ назовем блочно-единичной порядка~$P$.
\end{defin}

Далее через $0_{i,j}$ обозначаем нулевую $i{\times}j$-матрицу.

\begin{theorem}\label{th3}
Матрица алгебраической булевой вектор-функции ${C:\mB^m  \times
\mB^n \to \mB^k}$ элементарными преобразованиями приводится к
каноническому виду
\begin{equation}\notag
C=
\begin{tabular}{||m{1.2cm}|m{1.2cm}|c|c|m{2.1cm}||}
\multirow{2}{1.2cm}{ \centering $E_P$}&
\multirow{4}{1.2cm}{\centering
$0_{k,m-P}$} & \multicolumn{3}{c||}{ $0_{P-Q,n}$ } \\
\hhline{||~|~|-|-|-||} & & \centering $E_Q$ & \centering
$0_{Q,R}$ & \multirow{2}{2.1cm}{\centering $0_{Q+R,n-Q-R}$} \\
\hhline{||-|~|-|-|~||} \multirow{2}{1.2cm}{\centering $0_{k-P,P}$} &
{} & \centering
$0_{R,Q}$ & \centering $E_R$ & \\
\hhline{||~|~|-|-|-||} & & \multicolumn{3}{c||}{$0_{k-P-R,n}$}
\end{tabular}\,.
\end{equation}
Количество классов эквивалентности в $\mB^m$ относительно $C$ равно
$2^{P - Q}$.
\end{theorem}

\begin{proof} Матрицу $C$ представим в виде $C=||A|B||$,
где $A$ и $B$ -- матрицы индуцированных отображений $A:\mB^m  \to
B^k$, $B:\mB^n  \to \mB^k$, так что $ C({\mbf{v}},{\mbf{w}}) =
A{\mbf{v}} \oplus B{\mbf{w}}$ (сложение векторов над полем $\mB$).
Пусть $\rk A = P \leqslant \min (m,k)$. По лемме~\ref{lem11}
приводим $A$ к блочно-единичной матрице порядка~$P$:
$$
C = \begin{array}{||c|m{1.2cm}|c||}
   E_P & \multirow{2}{1.2cm}{\centering $0_{k,m-P}$}& B' \\
   \hhline{||-|~|-||}
   0_{k-P,P} & {}  & B''
 \end{array}\,.
$$
Пусть $\rk B'' = R$. Не затрагивая первые $P$ строк и первые $m$
столбцов, по той же лемме приводим $B''$ к блочно-единичной порядка
$R$, после чего с помощью строк с номерами ${P + 1, \ldots,}{P + R}$
обнуляем ту часть $B'$, которая заключена в столбцах ${m + 1,
\ldots,}{m + R}$. Получаем
\begin{equation}\notag
C=\begin{array}{||c|c|c|c||}
\begin{tabular}{c} $E_P$ \end{tabular}  &
\multirow{3}{1.3cm} {\centering $0_{k,m-P}$ } &  0_{P,R} &    B'''   \\
\hhline{||-|~|-|-||} \multirow{2}{1.3cm}{\centering $0_{k-P,P}$ } &
{} &  E_R &
\multirow{2}{2cm}{\centering  $0_{k-P,n-R}$ }   \\
\hhline{||~|~|-|~||}
& &  0_{k-P-R,R} &  {}  \\
\end{array}\, .
\end{equation}
Пусть $\rk B''' = Q$. Не затрагивая столбцов ${P + 1, \ldots ,}{m +
R}$ и строк ${P + 1, \ldots , k}$, приводим матрицу $B'''$ к
блочно-единичной порядка $Q$. Если при этом понадобилась
перестановка строк внутри номеров ${1, \ldots ,P}$, то такой же
перестановкой в столбцах ${1, \ldots ,P}$ восстановим
блочно-единичную структуру $A$:
\begin{equation}\notag
C=
\begin{array}{||m{1.2cm}|m{1.2cm}|m{1.6cm}|c|m{1.6cm}||}
\multirow{2}{1.2cm} {\centering $E_P$} &
\multirow{4}{1.2cm}{\centering $0_{k,m-P}$} &
\multirow{2}{1.6cm}{\centering $0_{P,R}$} &
E_Q   &   \multirow{4}{1.6cm}{\centering $0_{k,n-R-Q}$}  \\
\hhline{||~|~|~|-|~||}
& & &   0_{P-Q,Q}  &   {}  \\
\hhline{||-|~|-|-|~||} \multirow{2}{1.2cm}{\centering $0_{k-P,P}$} &
{} & \centering $E_R$
&   \multirow{2}{1.2cm}{ \centering $0_{k-P,Q}$}  &     \\
\hhline{||~|~|-|~|~||}
& & \centering  $0_{k-P-R,R}$ &   {}   &   \\
\end{array}\, .
\end{equation}
Отсюда, пользуясь лишь перестановкой строк и столбцов (внутри
группы), получаем для~$C$ канонический вид из утверждения теоремы.
Первое утверждение доказано.

По лемме~\ref{lem9}, не изменяя количества классов эквивалентности,
можем отбросить все столбцы и строки с номерами, большими, чем
$P-Q$. Получаем матрицу $E_{P - Q}$ тождественного отображения,
причем все аргументы относятся к первой группе. Поэтому каждый класс
эквивалентности содержит ровно один элемент, и таких классов $2^{P -
Q}$. Теорема доказана.
\end{proof}

Отметим, что для $\mB$-линейного отображения $C=A{\times}B$ числа
$P,Q,R$ представляют собой инварианты: $P=\rk A$, $Q+R=\rk B$,
$P+R=\rk C$, $Q=\dim (\mathop {\rm Im}\nolimits A \cap \mathop {\rm
Im}\nolimits B)$. В частности, все доказанные выше утверждения можно
суммировать в следующей теореме, вытекающей из линейности
алгебраической булевой вектор-функции и определения классов
эквивалентности. При этом канонический вид матрицы получается
подходящим выбором базисов в прообразе и образе.

\begin{theorem}
Пусть $C: U \to Z$ линейное отображение $($над произвольным полем$)$
и пусть $U=V \dot+ W$ $($прямая сумма$).$ Положим $v_1 \sim v_2$
$(v_1,v_2 \in V)$, если для некоторых $w_1,w_2 \in W$
$C(v_1+w_1)=C(v_2+w_2)$. Тогда:

$(1)$ класс эквивалентности нуля $V_0$ есть подпространство в $V$;

$(2)$ для любого $v\in V$ его класс эквивалентности есть $v + V_0 $,
в частности, множество классов эквивалентности изоморфно
фактор-пространству $V/V_0 $;

$(3)$ размерность $d = \dim V/V_0$ равна размерности
фактор-пространства $C(V)/Z_0 $, где $Z_0  = C(V) \cap C(W)$;

$(4)$ если поле конечно и состоит из $g$ элементов, то количество
классов эквивалентности равно $g^d$.
\end{theorem}

Интересно проанализировать сходство и различие с постановкой задачи
и результатами работ \cite{Comes1}, в которых также фактически
используется вычисление размерности некоторых фактор-пространств над
полем $\mB=\mathbb{Z}_2$.

\section{Приложение к классическим задачам}\label{sec4}
\subsection{Случай Чаплыгина -- Сретенского} Уравнения движения
гиростата в однородном поле имеют вид
\begin{equation}\label{eq4_1}
\begin{array}{l}
\displaystyle{\mbf{I}\frac{d{\boldsymbol
\omega}}{dt}=(\mbf{I}{\boldsymbol \omega}+{\boldsymbol
\lambda})\times {\boldsymbol \omega}+\mbf{r} \times {\boldsymbol
\alpha},}\qquad \displaystyle{\frac{d{\boldsymbol \alpha}}{dt}=
{\boldsymbol \alpha}\times {\boldsymbol \omega}.}
\end{array}
\end{equation}
В случае Чаплыгина~--~Сретенского предполагается, что после введения
безразмерных переменных получено ${\mbf{I}=\mathop{\rm
diag}\nolimits \{4,4,1\}}$, ${\bf r} = (1,0,0)$, ${\boldsymbol
\lambda}=(0,0,\lambda)$ и
\begin{equation}\label{eq4_2}
| {\boldsymbol \alpha}| = 1.
\end{equation}
Последнее соотношение (геометрический интеграл) рассматриваем как
определение фазового пространства $\bR^3 \times S^2 \subset {\bR}^6$
системы (\ref{eq4_1}). Общие интегралы таковы
\begin{equation}\notag
\begin{array}{l}
H = 2(\omega _1^2  + \omega _2^2 ) + \ds{\frac{1} {2}}\omega _3^2 -
\alpha _1, \qquad  G = 4(\omega _1 \alpha _1  + \omega _2 \alpha _2
) + (\omega _3 + \lambda )\alpha _3.
\end{array}
\end{equation}
На любом уровне интеграла $G$ система (\ref{eq4_1}) гамильтонова с
двумя степенями свободы с гамильтонианом $H$. На нулевом уровне $G =
0$ существует дополнительный интеграл, обобщающий интеграл
С.А.\,Чаплыгина и найденный Л.Н.\,Сре\-тен\-ским в работе
\cite{Sret}:
\begin{equation}\notag
K = 2(\omega _3  - \lambda )(\omega _1^2  + \omega _2^2 ) + 2\omega
_1 \alpha _3
\end{equation}
(множитель 2 введен для удобства). В частности, связные компоненты
регулярных интегральных многообразий индуцированной системы являются
двумерными торами~${\mbf{T}}^2$.

Полный топологический анализ этой системы, включающий описание всех
бифуркаций и их последовательностей вдоль непрерывных путей на
плоскости констант интегралов, выполнен в \cite{Kh831,Kh84,KhBk}.
Изображение бифуркаций в виде графов Фоменко вдоль прямых постоянной
энергии дано А.А.\,Ошемковым \cite{Oshem2,Oshem1}, вычисление
числовых характеристик, определяющих классы траекторной
эквивалентности выполнено в работе \cite{Orel}.

Незначительно модифицируя разделение переменных, указанное
Л.Н.\,Сре\-тен\-ским, введем переменные $s_1 ,s_2$, полагая
\begin{equation}\notag
\omega _1^2  + \omega _2^2  =  - \ds{\frac{s_1 s_2 }{4}},\quad
\omega _3 = s_1  + s_2  - \lambda \quad \quad (s_1  \geqslant s_2 ).
\end{equation}
Получим \cite{KhBk}
\begin{equation}\notag
2(s_2  - s_1 ){\ds{\frac{ds_1 }{dt}}} = \sqrt {V(s_1 )} ,\quad 2(s_2
- s_1 )\ds{\frac{ds_2 }{dt}} = \sqrt {V(s_2)} ,
\end{equation}
где
\begin{eqnarray}
& V(w) =  - W(w)W_* (w), \label{eq4_4} \\
& W(w) = w(w - \lambda )^2  - 2(h + 1)w - 2k, \qquad W_* (w) = w(w -
\lambda )^2  - 2(h - 1)w - 2k, \notag
\end{eqnarray}
а $h,k$ -- постоянные интегралов $H,K$ соответственно. Обозначим
\begin{equation}\label{eq4_5}
\begin{array}{l}
R_{11}  = \sqrt { - W(s_1 )} ,\quad R_{12}  = \sqrt {W_* (s_1
)},\quad  R_{21}  = \sqrt {W(s_2 )} ,\quad R_{22}  = \sqrt { - W_*
(s_2 )}.
\end{array}
\end{equation}
Алгебраическое решение задачи имеет вид
\begin{equation}\label{eq4_6}
\begin{array}{lll}
\omega _1  = \ds{\frac{1}{8}}(R_{11} R_{22} + R_{12} R_{21}),&
\omega _2  = \ds{\frac{1}{8}}(R_{12} R_{22}  - R_{11} R_{21}), & \omega _3  = s_1  + s_2  - \lambda,\\[4mm]
\alpha _1= 1 -\ds{\frac{{R_{11}^2  + R_{22}^2 }}{{2(s_1 - s_2 )}}},
& \alpha _2  = - \ds{\frac{R_{11} R_{12}  + R_{21} R_{22}} {2(s_1 -
s_2 )}} ,& \alpha _3  = \ds{\frac{R_{11} R_{22}  - R_{12} R_{21} }
{2(s_1 - s_2 )}}.
\end{array}
\end{equation}
Согласно лемме~\ref{lem7}, здесь в качестве базисных радикалов можно
выбрать сами радикалы (\ref{eq4_5}) без дальнейшего разложения их на
множители, а в качестве максимального многочлена выступает
многочлен~$V$, заданный в (\ref{eq4_4}). Он имеет кратный корень в
случаях
\begin{equation}\label{eq4_7}
\begin{array}{l}
k = 0,  \\
\Delta \phantom{ _*} = 27\left[k + \ds{\frac{2} {3}} \bigl(h +
1\bigr)\lambda - \ds{\frac{\lambda ^3 }{27}} \right]^2  - 8\bigl(h +
1 + \ds{\frac{\lambda ^2 }{6}}
\bigr)^3  = 0,  \\
\Delta _*  = 27\left[k + \ds{\frac{2} {3}}\bigl(h - 1\bigr)\lambda -
\ds{\frac{\lambda ^3 }{27}}\right]^2  - 8\bigl(h - 1 +
\ds{\frac{\lambda ^2 }{6}} \bigr)^3  = 0,
\end{array}
\end{equation}
что и дает уравнения разделяющего множества (содержащего в себе и
бифуркационную диаграмму \cite{KhBk}).

Полагая ${\mbf{u}} = \lsgn(R_{11}^2,R_{12}^2,R_{21}^2,R_{22}^2)$
(функцию $\lsgn$ к вектору применяем покомпонентно) и
${A}({\mbf{u}}) = \lsgn(R_{11}^2 R_{21}^2, R_{11}^2 R_{22}^2 ,
R_{12}^2 R_{21}^2 , R_{12}^2 R_{22}^2)$, получим булеву
вектор-функцию ${\mbf{z}} = {A}({\mbf{u}})$ для описания допустимых
областей в виде
\begin{equation}\label{eq4_8}
\begin{array}{l}
{A}:{\mbf{u}}\mapsto (u_1\oplus u_3,u_1\oplus u_4,u_2\oplus
u_3,u_2\oplus u_4).
\end{array}
\end{equation}
Условие вещественности выражений (\ref{eq4_6}) имеет теперь вид ${A}
(\mbf{u}) = 0000$, и, очевидно, $ {A}^{ - 1}(0000) = \{
0000,1111\}$. Получаем следующий критерий.
\begin{proposition} При отсутствии кратных корней
максимального многочлена область существования решений на плоскости
$(h,k)$ определяется условием: существует точка $(s_1 ,s_2 )$, для
которой все подкоренные выражения базисных радикалов либо
одновременно положительны, либо одновременно отрицательны. В точке,
отвечающей кратному корню, в этом условии нужно строгие неравенства
заменить на нестрогие.
\end{proposition}

Занумеруем области, на которые разделяющие кривые (\ref{eq4_7})
делят плоскость $(h,k)$. Здесь принята нумерация как в \cite{KhBk}.
Применительно к паре $\{ k = 0,\Delta  = 0\} $ области, отрезки
разделяющего множества и узловые точки показаны на
рис.~\ref{fig_sret1},{\it а} (здесь положение начала координат на
горизонтальной оси несущественно, поэтому ось $Ok$ показана
курсивом). Соответствующие объекты пары $\{k = 0,{\Delta_* = 0}\}$
имеют такое же обозначение, но снабженное звездочкой. Поэтому на
плоскости каждый объект (область, отрезок, точка) получает двойную
нумерацию. В результате имеем регулярные области, показанные на
рис.~\ref{fig_sret1},{\it б} для наиболее богатого случая $2\sqrt 3
< \lambda  < 2$.

Расположение корней многочленов $W,W_*$ приведено в
табл.~\ref{tabsr1}. Корни обозначаются $e_\gamma$ и $e^*_\gamma$ в
случае, когда у соответствующего многочлена три вещественных корня,
$e$ и $e^*$ в случае, когда вещественный корень единственный. Здесь
же указаны и промежутки для $s_1 ,s_2$, удовлетворяющие условию
предложения (см. аналогичный результат в \cite{KhBk}, полученный с
помощью анализа неравенств). Итак, в четырех последних областях
движений нет.

\begin{table}[ht]
\centering
\begin{tabular}{|c| c| c| c|}
\multicolumn{4}{r}{\fts{Таблица \myt\label{tabsr1}}}\\
\hline
\begin{tabular}{c}Номер\\области\end{tabular} &\begin{tabular}{c}Корни\end{tabular}
&\begin{tabular}{c}Область\\изменения $s_1$\end{tabular}&\begin{tabular}{c}Область\\изменения $s_2$\end{tabular} \\
\hline
$\ts{I-I}_*$&$e_1<e_1^*<e_2^*<e_2<0<e_3^*<e_3$&$[\,e_3^*,e_3\,]$&$[\,e_1,e_1^*\,]\cup[\,e_2^*,e_2\,]$\\
\hline
$\ts{I-III}_*$&$e_1<e_2<0<e_1^*<e_2^*<e_3^*<e_3$&$[\,e_1^*,e_2^*\,]\cup[\,e_3^*,e_3\,]$&$[\,e_1,e_2\,]$\\
\hline
$\ts{I-IV}_*$&$e_1<e_2<0<e^*<e_3$&$[\,e^*,e_3\,]$&$[\,e_1,e_2\,]$\\
\hline
$\ts{II-II}_*$&$e_1<e_1^*<0<e_2<e_2^*<e_3^*<e_3$&$[\,e_2,e_2^*\,]\cup[\,e_3^*,e_3\,]$&$[\,e_1,e_1^*\,]$\\
\hline
$\ts{II-V}_*$&$e_1<e^*<0<e_2<e_3$&$[\,e_2,e_3\,]$&$[\,e_1,e^*\,]$\\
\hline
$\ts{IV-III}_*$&$0<e_1^*<e_2^*<e_3^*<e$&$[\,e_1^*,e_2^*\,]\cup[\,e_3^*,e\,]$&$\varnothing$\\
\hline $\ts{IV-IV}_*$&$0<e^*<e$&$[\,e^*,e\,]$&$\varnothing$\\ \hline
$\ts{III-IV}_*$&$0<e^*<e_1<e_2<e_3$&$[\,e^*,e_1\,]\cup[\,e_2,e_3\,]$&$\varnothing$\\
\hline $\ts{V-V}_*$&$e<e^*<0$&$\varnothing$&$[\,e,e^*\,]$\\
\hline
\end{tabular}
\end{table}

\begin{figure}[ht]
\centering
\includegraphics[width=120mm,keepaspectratio]{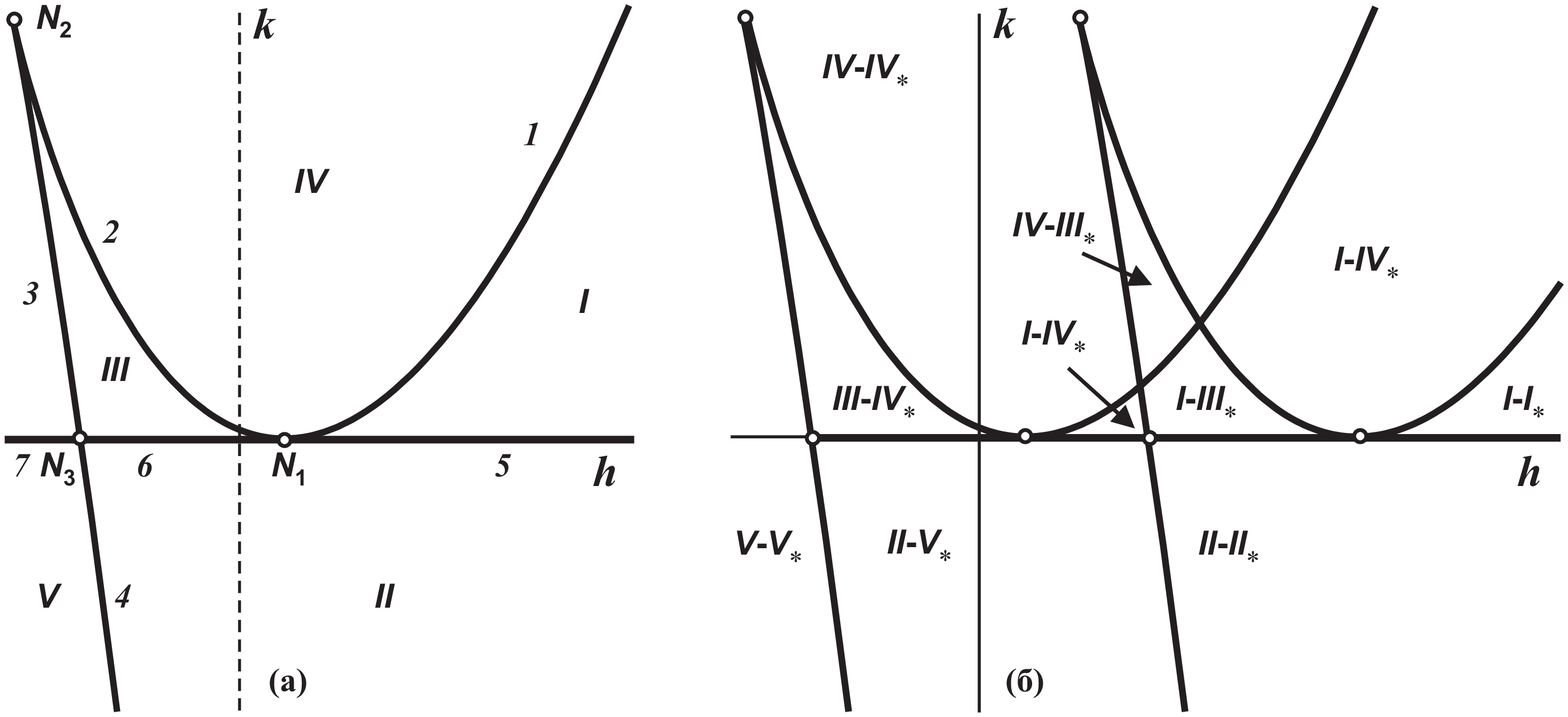}
\caption{Разделяющие кривые и кодировка областей.}\label{fig_sret1}
\end{figure}

Для вычисления количества связных компонент интегральных
многообразий составим булеву функцию $\mbf{z} =
{C}(\mbf{v},\mbf{w})$. В переменных $\mbf{u}$ она совпадает с
функцией (\ref{eq4_8}), но теперь ${\mbf{u}} =
\lsgn(R_{11},R_{12},R_{21},R_{22})$ и ${C}({\mbf{u}}) = \lsgn(R_{11}
R_{21}, R_{11} R_{22} , R_{12} R_{21} , R_{12} R_{22})$, то есть
квадраты радикалов (подкоренные выражения) заменяются на сами
радикалы. В табл.~\ref{tabsr2} представлены, во-первых, разбивка
радикалов (а значит, и аргументов $u_i$) на группы в зависимости от
области (номер области из табл.~\ref{tabsr1}) и варианта изменения
разделенных переменных (при наличии двух прямоугольников вариант
обозначен буквами {\it a,b}) и, во-вторых, результаты вычисления
количества классов эквивалентности (связных компонент интегрального
многообразия, накрывающих выбранный прямоугольник). В данной задаче
это количество несложно вычислить и непосредственно, без
преобразований матрицы, так как двоичная таблица булевой
вектор-функции содержит всего 16 строк.
\begin{table}[ht]
\centering
\begin{tabular}{|c| c| c| c| c|}
\multicolumn{5}{r}{\fts{Таблица \myt\label{tabsr2}}}\\
\hline
\begin{tabular}{c}Номер\\области\end{tabular} &\begin{tabular}{c}Область\\изменения $(s_1,s_2)$\end{tabular}&\begin{tabular}{c}Первая\\группа\end{tabular}
&\begin{tabular}{c}Вторая\\группа\end{tabular}&\begin{tabular}{c}Количество\\классов\end{tabular} \\
\hline
$\ts{I-I}_*,a$&$[\,e_3^*,e_3\,]\times[\,e_1,e_1^*\,]$&$\varnothing$&$R_{11},R_{12},R_{21},R_{22}$&1\\
\hline
$\ts{I-I}_*,b$&$[\,e_3^*,e_3\,]\times[\,e_2^*,e_2\,]$&$\varnothing$&$R_{11},R_{12},R_{21},R_{22}$&1\\
\hline
$\ts{I-III}_*,a$&$[\,e_1^*,e_2^*\,]\times[\,e_1,e_2\,]$&$R_{11},R_{22}$&$R_{12},R_{21}$&2\\
\hline
$\ts{I-III}_*,b$&$[\,e_3^*,e_3\,]\times[\,e_1,e_2\,]$&$R_{22}$&$R_{11},R_{12},R_{21}$&1\\
\hline
$\ts{I-IV}_*$&$[\,e^*,e_3\,]\times[\,e_1,e_2\,]$&$R_{22}$&$R_{11},R_{12},R_{21}$&1\\
\hline
$\ts{II-II}_*,a$&$[\,e_2,e_2^*\,]\times[\,e_1,e_1^*\,]$&$\varnothing$&$R_{11},R_{12},R_{21},R_{22}$&1\\
\hline
$\ts{II-II}_*,b$&$[\,e_3^*,e_3\,]\times[\,e_1,e_1^*\,]$&$\varnothing$&$R_{11},R_{12},R_{21},R_{22}$&1\\
\hline
$\ts{II-V}_*$&$[\,e_2,e_3\,]\times[\,e_1,e^*\,]$&$R_{12}$&$R_{11},R_{21},R_{22}$&1\\
\hline
\end{tabular}
\end{table}

На первый взгляд, результат по количеству компонент представляется
неверным. Действительно, если есть хоть один радикал, не меняющий
знак вдоль любой траектории на данном интегральном многообразии, то,
выбирая различные его знаки в начальной точке, мы должны получить
различные компоненты. Но дело в том, что он не обязательно входит в
выражения фазовых переменных сам по себе. Если он всегда умножается
на радикал из второй группы, то вдоль траектории выражение изменит
знак. Для примера возьмем область $\ts{I-III}_*$, для которой авторы
работы \cite{Ouz}, несмотря на данные им разъяснения, сумели
привести ошибочный результат. Здесь область $\A(h,k)$ состоит из
двух прямоугольников. Отметим сразу же, что $z_4=z_1\oplus z_2\oplus
z_3$. По лемме \ref{lem3} значение $z_4$ можно отбросить. В варианте
$\ts{I-III}_*,a$ фиксированы знаки у $R_{11},R_{22}$. Это означает,
в частности, что никакими путями на торе нельзя изменить значение
величины $z_2=\lsgn(R_{11} R_{22})$. Но если это значение выбрано,
то остальные два значения $z_1, z_3$ могут отвечать четырем наборам
знаков, так как содержат независимые слагаемые по модулю~2 --
$\lsgn(R_{21}), \lsgn(R_{12})$. Другие четыре набора отвечают
инверсному значению переменной $z_2$, а всего функция ${C}$ имеет
восемь различных значений. То есть, в этом варианте~-- компонент
две. Формальный путь здесь таков. Аргумент $u_2=\lsgn(R_{12})$
входит только в переменную $z_3$. Поэтому $z_3$ и $u_2$ можно
отбросить (лемма \ref{lem4}). После этого оказывается, что аргумент
$u_3=\lsgn(R_{21})$ входит только в переменную $z_1$. По той же
лемме $z_1$ и $u_3$ можно отбросить. Остаются два аргумента
$v_1=u_1,v_2=u_4$ и одна компонента функции $z_2=v_1 \oplus v_2$. У
такой функции ровно два значения и нет аргументов второй группы.
Поэтому классов эквивалентности, очевидно, два. В варианте
$\ts{I-III}_*,b$ фиксирован знак только у $R_{22}$. Но зависящим от
этого радикала переменным $z_2,z_4$ можно придать любые значения,
независимо меняя знаки радикалов $R_{11}, R_{12}$, или, что то же
самое, меняя значения аргументов $\lsgn(R_{11}), \lsgn(R_{12})$. То
есть этот прямоугольник накрывается одной компонентой. Рассуждая
формально, имеем следующее. Поскольку переменная $z_4$ уже
отброшена, то, применяя лемму~\ref{lem4}, последовательно
отбрасываем пары $(z_3,u_2)$, $(z_1,u_3)$. Осталась одна переменная
$z_2=u_1\oplus u_4$. Но $u_1$ -- аргумент второй группы. Пару
$(z_2,u_1)$ тоже можно отбросить, получив <<функцию без компонент>>.
По замечанию \ref{rem2} класс эквивалентности один. В целом же
прообраз точки из области $\ts{I-III}_*$ содержит три тора.

В итоге для регулярных интегральных многообразий имеем следующую
сводку, которая, естественно, совпадает с результатами \cite{KhBk}:
$\mbf{T}^2$ в областях $\ts{I-IV}_*$, $\ts{II-V}_*$; $2\mbf{T}^2$ в
областях $\ts{I-I}_*$, $\ts{II-II}_*$; $3\mbf{T}^2$ в области
$\ts{I-III}_*$.

Рассмотрим на базе той же булевой функции критические случаи.
Алгоритм вычисления количества компонент связности критических
поверхностей не меняется -- важно правильно указать разбивку
радикалов на группы и тип самой связной компоненты. Результаты
сведены в табл.~\ref{tabsr3}. Для примера разберем более подробно
переходы из области $\ts{I-III}_*$ в лежащие рядом области, имеющие
один и тот же шифр $\ts{I-IV}_*$. В соответствии с принятыми
обозначениями переход вправо (рис.~\ref{fig_sret1},{\it б})
происходит через отрезок $\ts{I-2}_*$, а влево -- через отрезок
$\ts{I-3}_*$. При первом переходе корни $e_1^*,e_2^*$ сливаются в
кратный корень $e_0^*$, а корень $e_3^*$ трансформируется в
сохраняющийся далее единственный корень $e^*$. В строке
$\ts{I-2}_*,b$ ничего не изменилось по сравнению с $\ts{I-III}_*,b$,
тип компоненты -- ${\mbf{T}}^2 $ (на торе не возникает критических
точек). В строке $\ts{I-2}_*,a$ для расчета, по сути дела, также
ничего не изменилось по сравнению с $\ts{I-III}_*,a$, так как
радикалы $R_{12} ,R_{21} $, ранее менявшие знак и потому не
вызывавшие удвоения прообраза, теперь тождественно равны нулю и
также не вызывают удвоения прообраза. Следовательно, эти радикалы
относятся ко второй группе. Однако для этого варианта компонента
выродилась в $S^1$. В результате, интегральная поверхность на
отрезке $\ts{I-2}_*$ есть $2S^1 \cup {\mbf{T}}^2 $. При втором
переходе через отрезок $\ts{I-3}_*$ ситуация иная. Здесь корни
$e_2^*,e_3^*$ сливаются в кратный корень $e_0^*$, а корень $e_1^*$
трансформируется в становящийся затем (в области $\ts{I-IV}_*$)
единственным корень $e^*$. Формально получаем третью строку
табл.~\ref{tabsr3}, такую же, как вторая, но теперь кратный корень
лежит внутри отрезка осцилляции. При этом здесь интегральная
поверхность связна, как и в случае $\ts{I-2}_*,b$. Для перестройки
$3{\mbf{T}}^2 \to {\mbf{T}}^2$ это может быть только $(S^1 \vee S^1
\vee S^1 ) \times S^1 $. В \cite{KhBk} эта перестройка установлена
аналитически, с помощью доказательства удвоения прообраза середины
<<восьмерки>>, имеющейся на плоскости $(s_1 ,v_1)$ ввиду того, что
$e_0^* \in (e^*,e_3)$. При этом доказывается, что при подъеме в
фазовое пространство левая часть <<восьмерки>> накрывается двумя
компонентами, а правая -- одной компонентой дважды. Там же приведены
соответствующие иллюстрации. Теперь же мы получили эти результаты
формальным алгоритмом.

\begin{table}[ht]
\centering
\begin{tabular}{|c| c| c| c| c|}
\multicolumn{5}{r}{\fts{Таблица \myt\label{tabsr3}}}\\
\hline
\begin{tabular}{c}Номер\\сегмента\end{tabular} &\begin{tabular}{c}Область\\изменения $(s_1,s_2)$\end{tabular}&
\begin{tabular}{c}Первая\\группа\end{tabular}
&\begin{tabular}{c}Вторая\\группа\end{tabular}&\begin{tabular}{c}Количество\\классов\end{tabular} \\
\hline
$\ts{I-2}_*,a$&$\{e_0^*\}\times[e_1,e_2]$&$R_{11},R_{22}$&$R_{12},R_{21}$&2\\
\hline
$\ts{I-2}_*,b$&$[e^*,e_3]\times[e_1,e_2]$&$R_{22}$&$R_{11},R_{12},R_{21}$&1\\
\hline
$\ts{I-3}_*$&$[e^*,e_3]\times[e_1,e_2]$&$R_{22}$&$R_{11},R_{12},R_{21}$&1\\
\hline
$\ts{6-7}_*$&$[e_2,e_3]\times\{e_1=e^*=0\}$&$R_{12}$&$R_{11},R_{21},R_{22}$&1\\
\hline
$\ts{5-7}_*$&$[e_2=e^*=0,e_3]\times[e_1,e_2=e^*=0]$&$\varnothing$&$R_{11},R_{12},R_{21},R_{22}$&1\\
\hline
$\ts{5-6}_*,a$&$[e_2=e_1^*=0,e_2^*]\times[e_1,e_2=e_1^*=0]$&$\varnothing$&$R_{11},R_{12},R_{21},R_{22}$&1\\
\hline
$\ts{5-6}_*,b$&$[e_3,e_3^*]\times[e_1,e_2=e_1^*=0]$&$\varnothing$&$R_{11},R_{12},R_{21},R_{22}$&1\\
\hline
$\ts{5-5}_*,a$&$\{e_2=e_1^*=0\}\times[e_1,e_1^*]$&$\varnothing$&$R_{11},R_{12},R_{21},R_{22}$&1\\
\hline
$\ts{5-5}_*,b$&$[e_3,e_3^*]\times\{e_2=e_1^*=0\}$&$\varnothing$&$R_{11},R_{12},R_{21},R_{22}$&1\\
\hline
$\ts{5-5}_*,c$&$[e_3^*,e_3]\times[e_1,e_1^*]$&$\varnothing$&$R_{11},R_{12},R_{21},R_{22}$&1\\
\hline
\end{tabular}
\end{table}

Рассмотрим еще с этой же точки зрения луч $\{k = 0,h \geqslant 1\}$.
В принятом интервале изменения $\lambda $ он разбит на четыре
качественно различных промежутка:
$$
\begin{array}{ll}
\displaystyle{\ts{6-7}_*: h \in (- 1,  - 1 + \frac{\lambda ^2
}{2});} & \displaystyle{\ts{5-7}_*: h\in ( - 1 + \frac{\lambda
^2}{2}, 1);} \\[3mm]
\displaystyle{\ts{5-6}_*: h\in (1, 1 +
\frac{\lambda ^2 }{2});} & \displaystyle{\ts{5-5}_*: h \in (1 +
\frac{\lambda ^2 } {2},+\infty).}
\end{array}
$$
Необходимая информация приведена начиная с четвертой строки
табл.~\ref{tabsr3}. Как видно из последнего столбца, во всех случаях
связная область плоскости $(s_1,s_2)$ накрывается связной
интегральной поверхностью. Установим топологический тип этой
поверхности. В случаях $\ts{6-7}_*$, $\ts{5-5}_*,a$, $\ts{5-5}_*,b$
одна из переменных не изменяется, поэтому компонента есть $S^1$. В
случае $\ts{5-5}_*,c$ изменение переменных не встречает кратных
корней, поэтому компонента регулярна -- тор ${\mbf{T}}^2 $. В случае
$\ts{5-7}_*$ имеем одну компоненту с топологическим типом расслоения
над <<восьмеркой>> со слоем окружность. При переходе через этот
участок, согласно составленной ранее таблице регулярных
многообразий, один тор преобразуется в один. Следовательно,
поверхность в этом случае есть косое произведение <<восьмерки>> на
окружность (обозначим ее через $(S^1 \vee S^1 ) * S^1 $). Эта
бифуркация впервые была установлена в работе \cite{Kh832}, а в книге
\cite{KhBk} приведено строгое аналитическое доказательство с
разбором поведения всех фазовых траекторий. Приведенный здесь подход
дает необходимый результат напрямую. В случае $\ts{5-6}_*$ имеем две
компоненты. При этом обе имеют тип расслоения над <<восьмеркой>> со
слоем окружность. Но при переходе через этот участок, как было
установлено выше, три тора преобразуются в два. Поэтому одно из
расслоений тривиально $(S^1 \vee S^1 ) \times S^1 $, второе -- косое
произведение <<восьмерки>> на окружность $(S^1  \vee S^1)*S^1$.

В итоге для $\{k = 0,h \geqslant 1\}$ получаем следующую сводку
интегральных поверхностей: $S^1$ на участке $\ts{6-7}_*$;  $(S^1
\vee S^1 ) * S^1 $ на участке $\ts{5-7}_*$; объединение поверхностей
$(S^1\vee S^1)*S^1$ и $(S^1\vee S^1)\times S^1$ на участке
$\ts{5-6}_*$; $2S^1\cup{\mbf{T}}^2$ на участке $\ts{5-5}_*$.

\subsection{Случай Ковалевской. Достижимые области} В уравнениях
(\ref{eq4_1}) положим $\mbf{I}=\mathop{\rm diag}\nolimits
\{2,2,1\}$, $\mbf{r} = (1,0,0)$, ${\boldsymbol \lambda}=0$. Для
геометрического интеграла сохраним константу (\ref{eq4_2}). Общие
интегралы
\begin{equation}\label{eq4_9}
\begin{array}{l}
H = \omega _1^2  + \omega _2^2 + \ds{\frac{1}{2}}\omega _3^2 -
\alpha _1, \qquad
G = 2(\omega _1 \alpha _1  + \omega _2 \alpha _2 ) + \omega _3 \alpha _3, \\
K=(\omega_1^2-\omega^2_2+\alpha_1)^2+(2\omega_1\omega_2+\alpha_2)^2.
\end{array}
\end{equation}

Ограничения на уровни интеграла $G$ -- вполне интегрируемые
гамильтоновы системы с двумя степенями свободы. Их топологический
анализ, включающий построение бифуркационных диаграмм, определение
топологического типа регулярных интегральных многообразий и
критических интегральных поверхностей, описание бифуркаций вдоль
всех путей в допустимой области выполнен в \cite{Kh831,Kh832, KhBk}.
Инварианты Фоменко -- Цишанга и описание круговых молекул приведены
в \cite{BolFomRus} и \cite{BolRich}.

Следуя С.В.\,Ковалевской \cite{Kowa}, вводим комплексные переменные
\begin{equation}\notag
\begin{array}{ll}
   {w_1  = \omega _1  + \ri\; \omega _2 ,} & {w_2  = \omega _1  - \ri\;\omega _2 ,}\\
   {x_1  = \alpha _1  + \ri\; \alpha _2 ,} & {x_2  = \alpha _1  - \ri\;\alpha _2 .}
\end{array}
\end{equation}
Уравнения разделяются в переменных Ковалевской $s_1 ,s_2$, которые
определяются как корни квадратного уравнения
\begin{equation}\label{eq4_10}
(s - h)^2 - \frac{ 2R(w_1 ,w_2)} { (w_1 - w_2)^2} (s - h) - \frac{
R_1(w_1,w_2)} { (w_1 - w_2)^2}  = 0,
\end{equation}
где
\begin{equation}\notag
\begin{array}{l}
R_1 (w_1 ,w_2 ) =  - 2hw_1^2 w_2^2  - 4gw_1 w_2 (w_1  + w_2 ) - (1 - k)(w_1  + w_2 )^2  - 4g^2 ,  \\
R(w_1 ,w_2 ) =  - w_1^2 w_2^2  + 2hw_1 w_2  + 2g(w_1  + w_2 ) + 1 -
k.
\end{array}
\end{equation}
Наряду с этим используется также обозначение многочлена от одной
переменной
\begin{equation}\notag
R(w) =  - w^4  + 2hw^2  + 4gw + 1 - k.
\end{equation}
Разделенная система имеет вид
\begin{equation}\label{eq4_11}
(s_2  - s_1 )\frac{ds_1 } {dt} = \ri \sqrt {2S(s_1 )} ,\quad (s_2  -
s_1 ) \frac{ds_2 } {dt} =  - \ri \sqrt {2S(s_2 )} ,
\end{equation}
где
\begin{eqnarray}
& & S(s) = (s - h + \sqrt k )(s - h - \sqrt k )\varphi (s),\label{eq4_12}\\
& & \varphi (s) = s(s - h)^2  + (1 - k)s - 2g^2 . \notag
\end{eqnarray}
Считаем, что всегда
\begin{equation}\label{eq4_13}
s_1  \geqslant s_2 .
\end{equation}

Максимальный многочлен задачи -- многочлен (\ref{eq4_12}).
Бифуркационная диаграмма $\Sigma $ интегралов (\ref{eq4_9})
содержится в дискриминантном множестве $\tilde \Sigma$ многочлена
$S(s;{\mathbf{f}})$, ${\mathbf{f}} = (g,h,k)$. Допустимая область
есть $\mathop{\rm Im}\nolimits \mF$, где $\mF=G{\times}K{\times}H$,
так что $\Sigma = \tilde \Sigma \cap \mathop{\rm Im}\nolimits \mF$
есть часть $\tilde \Sigma$, отвечающая непустым интегральным
поверхностям, то есть, в свою очередь, непустым достижимым областям
на плоскости переменных $s_1 ,s_2$. Из уравнений $S(s;{\mathbf{f}})
= 0$ и $S'_s (s;{\mathbf{f}}) = 0$ получаем структуру $\tilde \Sigma
$ в виде объединения поверхностей:

1) $\pi _1 $: $k = 0$ (1-й класс Аппельрота);

2) $\pi _2 $: $h = 2g^2  - \sqrt k $ (2-й класс Аппельрота);

3) $\pi _3 $: $h = 2g^2  + \sqrt k $ (3-й класс Аппельрота);

4) поверхность кратных корней многочлена $\varphi (s;{\mathbf{f}})$
(4-й класс Аппельрота)
\begin{equation}\label{eq4_14}
h = s + \ds{\frac{g^2}{s^2}},\quad k = 1 - \ds{\frac{2g^2}{s}} +
\ds{\frac{g^4}{s^4}},
\end{equation}
которая естественным образом распадается на три участка ${\pi _4: {s
< 0}}$, ${\pi _5: {0 < s < s_0}}$,  ${\pi _6: {s > s_0}}$. Здесь
$s_0  = \root 3 \of {2g^2 } $ -- значение параметра в точках ребра
возврата.

\begin{figure}[ht]
\centering
\includegraphics[width=100mm,keepaspectratio]{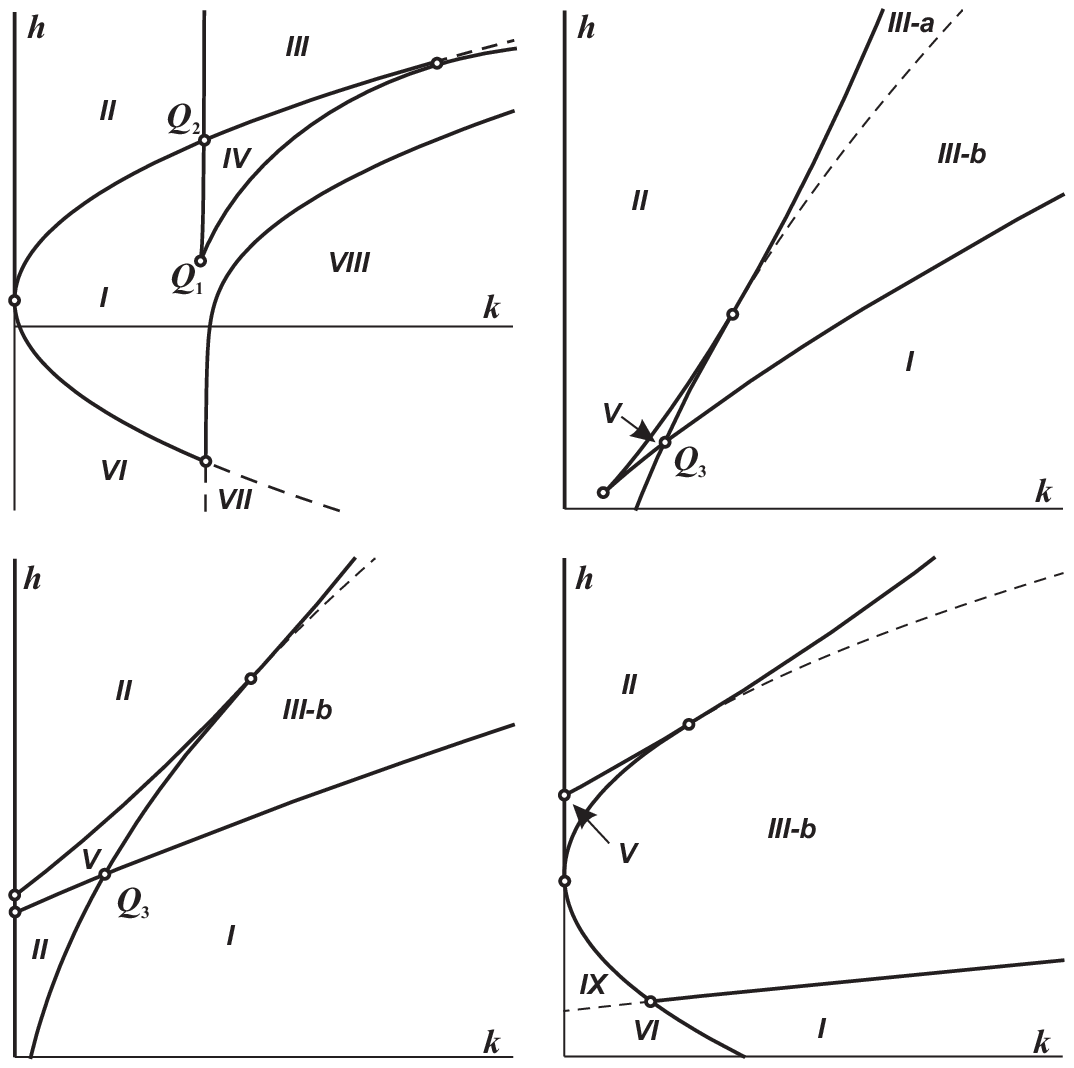}
\caption{Разделяющие множества и кодировка
областей.}\label{fig_kow1}
\end{figure}

Поскольку разрешимыми системами здесь являются ограничения на уровни
интеграла $G$, то естественно рассматривать сечения дискриминантной
поверхности плоскостями $g = {\rm{const}}$, так что все $\pi _i $
трактуются как плоские кривые, которые ниже мы называем
разделяющими. Как показано в \cite{Kh831}, различные типы сечений
(устойчивые по параметру $g$) имеют место в случаях: 1) $0 < g^2  <
1/2$; 2) $1/2 < g^2  < 4/(3\sqrt 3 )$; 3) $4/(3\sqrt 3 ) < g^2  <
1$; 4) $g^2  > 1$. В силу очевидной симметрии считаем в дальнейшем,
что
\begin{equation}\label{eq4_15}
g > 0.
\end{equation}

На рис.~\ref{fig_kow1} показаны ключевые участки сечений $\tilde
\Sigma $, поясняющие нумерацию областей $\ts{I}$ -- $\ts{IX}$, на
которые эта поверхность разбивает полупространство ${\bR}^3
\backslash \{ k < 0\}$. На этом же рисунке отмечены особые точки
$Q_1 - Q_3$, о которых речь пойдет позже.

Обозначения $\ts{III-a}$ и $\ts{III-b}$ введены для удобства
сопоставления с нумерацией областей, принятой после работ
\cite{Kh831,Kh832}. Как доказано еще Г.Г.\,Аппельротом \cite{Appel},
на участке разделяющей кривой между этими областями никаких
бифуркаций не происходит. Подчеркнем, что в данном контексте это
неважно.

Условимся о следующих обозначениях. Если многочлен $\varphi (s)$
имеет три различных вещественных корня, то обозначим их в порядке
возрастания
\begin{equation}\label{eq4_16}
e_1  < e_2  < e_3 .
\end{equation}
Если вещественный корень единственный, то это будет $e_3 $, а $e_1 $
и $e_2 $ будут комплексно сопряжены. Обозначим также остальные два
корня максимального многочлена через
\begin{equation}\label{eq4_17}
e_4  = h - \sqrt k  < e_5  = h + \sqrt k .
\end{equation}

Вначале решаем задачу определения промежутков изменения $s_1 ,s_2$.
Условий ${S(s_i)\leqslant 0}$, вытекающих непосредственно из
уравнений (\ref{eq4_11}), для существования вещественных решений
системы (\ref{eq4_2}), (\ref{eq4_9}) недостаточно. Однако на этом
этапе нет и необходимости использовать явные выражения всех фазовых
переменных через $s_1 ,s_2 $. Важную роль во всех исследованиях
случая Ковалевской играют проекции интегральных многообразий на
плоскость $(\omega _1 ,\omega _2 )$, поскольку именно через эти
переменные согласно (\ref{eq4_10}) определяются переменные
разделения. Движения, в которых $\omega _2 \equiv 0$, называются
исключительными \cite{Appel}, и именно они порождают поверхность
(\ref{eq4_14}). В частности, во всех остальных случаях множество
точек интегрального многообразия, в которых $\omega _2  = 0$, имеет
на этом многообразии меру нуль. Поэтому верно следующее утверждение:
в случае отсутствия кратных корней у максимального многочлена для
почти всех точек интегрального многообразия значения переменных
$\alpha_1, \alpha_2$ однозначно определяются по значениям переменных
$\omega _1, \omega_2, \omega_3 ,\alpha_3$ из уравнений первых
интегралов (\ref{eq4_9}). В то же время $w_1 ,w_2$ однозначно
связаны с $\omega _1 ,\omega _2$, и, как показано в
\cite{Appel,Ipat}, переменные $w_1 ,w_2 ,\omega_3 , \alpha_3$
полностью определены следующими двумя системами уравнений:
\begin{equation}
\displaystyle{\frac{\sqrt {e_\gamma}}{\sqrt{2}}} (w_1  - w_2
)R_{1\gamma } R_{2\gamma } = - e_\gamma ^2  + (w_1 w_2  + h)e_\gamma
+ (w_1  + w_2 )g \quad (\gamma  = 1 \div 3); \label{eq4_18}
\end{equation}
\begin{equation}
\begin{array}{l} \displaystyle{ \omega_3 w_1  + \alpha_3 =
\ds{\frac{\sqrt{R(w_1 )}} {s_2  - s_1 }}
(R_{24} R_{25}  + R_{14} R_{15} ),}  \\
\omega_3 w_2  + \alpha_3 = \ds{\frac{\sqrt{R(w_2 )}} {s_2  - s_1
}}(R_{24} R_{25}  - R_{14} R_{15} ).
\end{array}\label{eq4_19}
\end{equation}
Здесь $R_{i\gamma } = \sqrt{s_i-e_\gamma}$ ($i = 1,2$; $\gamma  =
1\div 5$) являются базисными радикалами.

\begin{remark}\label{remk1}
При выводе этих уравнений предполагается, что выбраны комплексно
сопряженные значения величин $\sqrt {R(w_1 )}$, $\sqrt {R(w_2 )}$, а
константы $\sqrt {e_\gamma}$ взяты с такими знаками, что
\begin{equation}\notag
2g = \sqrt {\mathstrut 2} \sqrt {\mathstrut e_1 } \sqrt {\mathstrut
e_2 } \sqrt {\mathstrut e_3 } .
\end{equation}
Согласно предположению $(\ref{eq4_15})$ и обозначению
$(\ref{eq4_16})$ можем считать всегда
\begin{equation}\notag
\sqrt {e_3 }  > 0,\quad \sqrt {e_1 } \sqrt {e_2 }  > 0.
\end{equation}
\end{remark}

Три уравнения (\ref{eq4_18}) линейно зависимы в силу связи величин
$e_\gamma$ с параметрами $h,g$. Однако их удобно использовать для
выражения переменных $w_1 ,w_2 $ и для исследования условий
существования решений. Для вещественных $e_\gamma$ правые части
(\ref{eq4_18}) вещественны, разность $w_1 - w_2 $ всегда чисто
мнимая. Поэтому условием разрешимости (\ref{eq4_18}) является
следующее свойство: для вещественных корней $e_\gamma  $ ($\gamma  =
1\div 3$) выполнены неравенства
\begin{equation}\notag
e_\gamma  R_{1\gamma }^2 R_{2\gamma }^2  > 0.
\end{equation}
Здесь уместно напомнить, что достаточно решить задачу для параметров
${\mathbf{f}}$, при которых кратные корни отсутствуют, в связи с чем
пишем строгие неравенства, соответствующие внутренней точке области
$\A({\mathbf{f}})$. По замечанию~\ref{remk1} система~(\ref{eq4_19})
разрешима тогда и только тогда, когда правые части взаимно
сопряжены, что равносильно условиям
\begin{equation}\notag
R_{14}^2 R_{15}^2  < 0,\quad R_{24}^2 R_{25}^2  > 0.
\end{equation}
Обозначим
\begin{equation}\label{eq4_20}
u_\gamma   = \lsgn R_{1\gamma }^2 ,\quad u_{5 + \gamma } = {\lsgn }
R_{2\gamma }^2 ,
\end{equation}
где $\gamma  = 1 \div 5$, если $\varphi (s)$ имеет три вещественных
корня, и $\gamma  = 3 \div 5$, если $e_1 ,e_2 $ -- комплексные.
Вводя компоненты
\begin{equation}\notag
z_\gamma   = u_\gamma   \oplus u_{5 + \gamma } \;(\gamma  \leqslant
3),\quad z_4  = u_4  \oplus u_5 ,\quad z_5  = u_9  \oplus u_{10} ,
\end{equation}
получаем алгебраическую булеву вектор-функцию ${\mathbf{z}} =
{A}({\mathbf{u}})$, где ${\mathbf{u}} \in \mB^{10}$, ${\mathbf{z}}
\in \mB^5$ или ${\mathbf{u}} \in \mB^6$, ${\mathbf{z}} \in \mB^3$ в
зависимости от количества вещественных корней $e_\gamma$. Матрица
$A$ в обоих случаях имеет максимальный ранг. Условия существования
внутренней точки множества $\A({\mathbf{f}})$ принимают вид
\begin{equation}\label{eq4_21}
{\mathbf{z}_0}={A} ({\mathbf{u}}) = \left\{
\begin{array}{lll}
{11110,} & {e_1 ,e_2  \in {\bR},} & {e_1  > 0,} \\
{00110,} & {e_1 ,e_2  \in {\bR},} & {e_2  < 0,} \\
{110,} & {e_1 ,e_2  \notin {\bR}.} & {}
\end{array} \right.
\end{equation}
Размерность пространства решений равна 5 в первых двух случаях и 3
-- в последнем. Количество решений, соответственно, 32 и 8. Поэтому
систем неравенств для анализа слишком много. Вспомним, однако, что
здесь дополнительно можно учесть условия, заведомо существующие в
силу (\ref{eq4_13}), (\ref{eq4_16}) и (\ref{eq4_17}). В случае
одного вещественного корня к функции ${A}({\mathbf{u}})$ припишем
компоненты
\begin{equation}\label{eq4_22}
z_6  = u_3  \to u_8 ,\quad z_7  = u_4  \to u_9 ,\quad z_8  = u_5 \to
u_{10} ,
\end{equation}
согласно неравенству $s_1  > s_2$ и компоненты
\begin{equation}\label{eq4_23}
z_9  = u_4  \to u_5 ,\quad z_{10}  = u_9  \to u_{10} ,
\end{equation}
согласно неравенству $e_4 < e_5 $. Если вещественных корней три, то
к (\ref{eq4_22}) добавляются
\begin{equation}\notag
z_{11}  = u_1  \to u_6 ,\quad z_{12}  = u_2  \to u_7 ,
\end{equation}
а к (\ref{eq4_23}) -- компоненты-импликации, соответствующие порядку
$e_1 < e_2 < e_3$:
\begin{equation}\notag
\begin{array}{l}
z_{13}  = u_1  \to u_2 ,\quad z_{14}  = u_6  \to u_7 ,  \\
z_{15}  = u_2  \to u_3 ,\quad z_{16}  = u_7  \to u_8 .
\end{array}
\end{equation}
Потребовав для всех компонент-импликаций $z_j  = 1$, то есть,
приписав справа к значению ${\mathbf{z}_0}$ в (\ref{eq4_21})
соответствующее количество единиц, найдем, что во всех случаях
прообраз состоит ровно из одной точки:
\begin{equation}\notag
{A}^{ - 1} ({\mathbf{z}_0}) = \left\{
\begin{array}{lll}
{\{ 0000111111\} ,}  & {e_1 ,e_2  \in {\bR},}  & {e_1  > 0,}   \\
{\{ 0000100111\} ,}  & {e_1 ,e_2  \in {\bR},}  & {e_2  < 0,}   \\
{\{ 001111\} ,}  & {e_1 ,e_2  \notin {\bR}.}  & {}
\end{array} \right.
\end{equation}
Отсюда сразу же определяются промежутки осцилляции $s_1 ,s_2$:
\begin{eqnarray}
& & e_1 ,e_2  \notin {\bR} \Rightarrow \left\{
\begin{array}{l}
   {s_1  \in [\max \{ e_3 ,e_4 \} ,e_5 ]}   \\
   {s_2  \in [ - \infty ,\min \{ e_3 ,e_4 \} ]}
\end{array} \right.,\label{eq4_24}\\
& & e_1 ,e_2  \in {\bR},\;e_1  > 0 \Rightarrow \left\{
\begin{array}{l}
   {s_1  \in [\max \{ e_1 ,e_2 ,e_3 ,e_4 \} ,e_5 ]}   \\
   {s_2  \in [ - \infty ,\min \{ e_1 ,e_2 ,e_3 ,e_4 ,e_5 \} ]}
\end{array} \right.,\label{eq4_25} \\
& & e_1 ,e_2  \in {\bR},\;e_2  < 0 \Rightarrow \left\{
\begin{array}{l}
   {s_1  \in [\max \{ e_1 ,e_2 ,e_3 ,e_4 \} ,e_5 ]}   \\
   {s_2  \in [\max \{ e_1 ,e_2 \} ,\min \{ e_3 ,e_4 ,e_5 \} ]}
\end{array} \right..\label{eq4_26}
\end{eqnarray}
С учетом расположения корней в областях \ts{I} -- \ts{IX}, получаем
классификацию достижимых областей согласно табл.~\ref{tabkow1}
(обозначения в последнем столбце поясним позже).
\begin{table}[ht]
\centering
\begin{tabular}{|c| c| c| c| c|}
\multicolumn{5}{r}{\fts{Таблица \myt\label{tabkow1}}}\\
\hline
\begin{tabular}{c}Номер\\области\end{tabular} &\begin{tabular}{c}Корни\end{tabular}
&\begin{tabular}{c}Область\\изменения $s_1$\end{tabular}&\begin{tabular}{c}Область\\изменения $s_2$\end{tabular}
&\begin{tabular}{c}Первая\\группа\end{tabular} \\
\hline
\ts{I}& $e_4<e_3<e_5$&$[\,e_3,e_5\,]$&$[-\infty,e_4\,]$&$124678^*$\\
\hline
\ts{II}& $e_3<e_4<e_5$&$[\,e_4,e_5\,]$&$[-\infty,e_3\,]$&$123679^*$\\
\hline
\ts{III-a}&$0<e_1<e_4<e_2<e_3<e_5$&$[\,e_3,e_5\,]$&$[-\infty,e_1\,]$&124789\\
\hline
\ts{III-b}&$0<e_1<e_2<e_4<e_3<e_5$&$[\,e_3,e_5\,]$&$[-\infty,e_1\,]$&124789\\
\hline
\ts{IV}&$e_4<e_1<e_2<e_3<e_5 \, (e_1>0)$&$[\,e_3,e_5\,]$&$[-\infty,e_4\,]$&124678\\
\hline \ts{V}&$0<e_1<e_2<e_3<e_4<e_5$ &
$[\,e_4,e_5\,]$&$[-\infty,e_1\,]$&123789\\
\hline \ts{VI}&$e_4<e_5<e_3$&$\varnothing$&$[-\infty,e_4\,]$&{}\\
\hline
\ts{VII}&$e_4<e_1<e_2<e_5<e_3 \, (e_2<0)$&$\varnothing$&$\varnothing$&{}\\
\hline \ts{VIII}&$e_4<e_1<e_2<0<e_3<e_5$&$[\,e_3,e_5\,]$&$\varnothing$&{}\\
\hline \ts{IX}&$0<e_1<e_2<e_4<e_5<e_3$&$\varnothing$&$[-\infty,e_1\,]$&{}\\
\hline
\end{tabular}
\end{table}
Таким образом, в областях \ts{VII} и \ts{VIII}, отвечающих случаю
(\ref{eq4_26}), движения невозможны, так как $e_4<\max\{e_1,e_2\}$.
В области \ts{VI} имеем $\max\{e_3,e_4\}>e_5$, что противоречит
(\ref{eq4_24}), а в области \ts{IX} неравенство $e_5< e_4$
противоречит (\ref{eq4_25}). Поэтому допустимая область $\mathop{\rm
Im}\nolimits \mathcal{F}$ в пространстве постоянных
${\mathbf{f}=(g,h,k)}$ состоит из областей \ts{I} -- \ts{V} и
примыкающих к ним участков разделяющего множества, которые,
следовательно, и формируют бифуркационную диаграмму $\Sigma$. Мы
снова обошлись без аналитических исследований, сразу получив
результаты работ \cite{Appel,Ipat} о допустимых областях и
\cite{Kh831,Kh832} о бифуркационной диаграмме.

Условия $s_1 = {\rm{const}}$ и $s_2  = {\rm{const}}$ при
фиксированных постоянных ${\mbf{f}}$, в силу определения
(\ref{eq4_10}), можно рассматривать как уравнения координатной сети
на плоскости $(\omega _1 ,\omega _2 )$. Это впервые сделал
Н.Е.\,Жуковский в работе \cite{Zhuk}. Он изучил свойства этой сети и
указал достижимые области $\A_\omega(\mbf{f})$ на плоскости $(\omega
_1 ,\omega _2 )$, то есть, фактически, проекции на эту плоскость
интегральных многообразий. Более подробное аналитическое
исследование таких областей и их геометрическое представление даны в
работах \cite{Appel,Ipat}. Оказалось, что при отсутствии кратных
корней связные компоненты множества $\A_\omega (\mbf{f})$ с
точностью до диффеоморфизма бывают только двух видов --
прямоугольник и кольцо. В работе \cite{Kh831} показано, что в
прямоугольник проецируется один тор, а в кольцо -- два. Из этого
сделан вывод о количестве торов в составе $\mathcal{F}_{\mbf{f}} $ в
облас\-тях~\mbox{\ts{I}\,--\,\ts{V}}.
\begin{remark}\label{remk2}
Как показано в {\rm \cite{Kh831}}, каждая внутренняя точка области
$\A_\omega (\mbf{f})$ накрывается четырьмя точками
$\mathcal{F}_\mbf{f}$, прообразом граничной точки кольца и точки на
стороне прямоугольника служат две точки интегрального многообразия,
а прообраз вершины прямоугольника состоит ровно из одной точки
$($подробности см. в {\rm \cite{KhBk}}$).$ Из последнего свойства с
необходимостью следует, что прообраз прямоугольника имеет одну
связную компоненту, то есть состоит из одного тора. Что касается
кольца, ответ не столь очевиден. Такую же картину по количеству
прообразов можно получить при двойном накрытии кольца одним тором,
положив, например, в полярных координатах $(\rho ,\theta )$
плоскости
\begin{equation}\notag
\rho  = 2 + \cos \varphi _1 ,\quad \theta  = 2\varphi _2 \quad
(\varphi _1 ,\varphi _2 \, {\rm modd}\, 2\pi ).
\end{equation}
Поэтому строгое доказательство утверждений о количестве связных
компонент дал анализ проекций интегральных многообразий на сферу
Пуассона. Ниже мы покажем, как доказываются утверждения {\rm
\cite{Kh831}} методом данной работы.
\end{remark}

В качестве дополнительной иллюстрации работы метода установим
количество компонент в областях Жуковского $\A_\omega (\mbf{f})$.
Решение системы (\ref{eq4_18}) имеет вид \cite{Kowa,Appel}
\begin{equation}\label{eq4_27}
\displaystyle{ \omega _1  = -\frac{\sum
\displaystyle{\frac{\sqrt{e_\delta}\sqrt{e_\nu}}{\varphi'(e_\gamma)}}R_{1
\gamma}R_{2 \gamma}}{\sqrt{2}\sum
\displaystyle{\frac{\sqrt{e_\gamma}}{\varphi'(e_\gamma)}}R_{1
\gamma}R_{2 \gamma} },}\qquad \displaystyle{ \omega _2  =
\frac{\ri}{\sqrt{2}\sum
\displaystyle{\frac{\sqrt{e_\gamma}}{\varphi'(e_\gamma)}}R_{1
\gamma}R_{2 \gamma} }.}
\end{equation}
Здесь принимается введенная еще Ковалевской договоренность об
упрощенном обозначении: суммирование ведется по $\gamma $ в пределах
от 1 до 3, в слагаемом с номером $\gamma $ индексы $\delta ,\nu $
означают два числа из набора (1,2,3), отличные от $\gamma$. Итак,
точка $\Omega (\omega _1 ,\omega _2 )$ однозначно определяется
знаками трех произведений базисных радикалов
\begin{equation}\label{eq4_28}
P_\gamma   = R_{1\gamma } R_{2\gamma }.
\end{equation}
В частности, каждой точке $(s_1 ,s_2 ) \in {\mathop{\rm
Int}\nolimits} \A({\mbf{f}})$ соответствует 8 точек из
$\A_\omega(\mbf{f})$.

Отбросим квадраты в определении булевых аргументов (\ref{eq4_20}):
\begin{equation}\label{eq4_29}
u_\gamma   = {\lsgn } R_{1\gamma } ,\quad u_{5 + \gamma }  = {\lsgn
} R_{2\gamma } \quad (\gamma = 1 \div 5).
\end{equation}
Из перечня достижимых областей (табл.~\ref{tabkow1}) видно, что к
первой группе по переменной $s_1$ относятся булевы переменные $u_1
,u_2 ,u_3$, если $s_1 $ осциллирует в $[\,e_4 ,e_5 \,]$, и $u_1 ,u_2
,u_4$, если $s_1$ осциллирует в $[\,e_3 ,e_5 \,]$. Отмечаем их
номера в последнем столбце той же таблицы.

Переменная $s_2$ всегда изменяется в бесконечной полуполосе, то есть
периодически, в конечные моменты времени, <<отражается>> от значения
$- \infty$. При прохождении этого значения формально следует считать
одновременно меняющими знак все базисные радикалы, содержащие $s_2$,
однако, это не отражает существа дела, поскольку в зависимостях вида
(\ref{eq4_27}) для подобных случаев особенность в бесконечности
всегда устранима. Поэтому поступим следующим образом. Обозначим
\begin{equation}\label{eq4_30}
\begin{array}{l}
\displaystyle{ R^*_{2 \gamma} = \sqrt { \frac{s_2  - e_\gamma  }
{s_2  - e_5 }} \quad (\gamma  = 1 \div 4),}  \qquad \displaystyle{
R^*_{2 5}  = \frac{1}{\sqrt {s_2 - e_5 }} .}
\end{array}
\end{equation}
Переопределим переменные, отвечающие знакам $R_{2\gamma } $ в
(\ref{eq4_29}), полагая
\begin{equation}\notag
\begin{array}{l}
\displaystyle{ u_{5 + \gamma }  = {\lsgn } R^*_{2 \gamma} \quad
(\gamma  = 1 \div 5).}
\end{array}
\end{equation}
Поскольку переменная $s_2 $ никогда не принимает значение $e_5  =
\max \{ e_\gamma  \} $, то при ее колебаниях на любом промежутке
вида $[ - \infty ,e_\delta  \,]$, где $e_\delta   = \min \{ e_\gamma
\} $, меняющими знак надо считать ровно две из введенных величин, а
именно,
\begin{equation}\notag
\displaystyle{\sqrt {\frac{s_2  - e_\delta  } {s_2  - e_5 }} ,\quad
\frac{1} {\sqrt {s_2  - e_5 }},}
\end{equation}
то есть к первой группе не относятся переменные $u_{5 + \delta}$ и
$u_{10}$. Итак, согласно таблице достижимых областей, к первой
группе по переменной $s_2$ относим булевы переменные $u_7 ,u_8 ,u_9
$, если $s_2 $ осциллирует в $[ - \infty ,e_1 \,]$, $u_6 ,u_7 ,u_9
$, если $s_2$ осциллирует в $[ - \infty ,e_3 \,]$ и $u_6 ,u_7 ,u_8
$, если $s_2 $ осциллирует в $[ - \infty ,e_4 \,]$. И эти номера
отмечаем в последнем столбце табл.~\ref{tabkow1}. Строки индексов,
помеченные звездочкой, соответствуют случаям, когда согласно
замечанию \ref{rem1} при комплексно сопряженных $e_1,e_2$
рассматриваются только наборы аргументов, у которых
\begin{equation}\label{eq4_31}
u_1  \oplus u_2=0, \qquad u_6 \oplus u_7=0.
\end{equation}
Теперь в зависимостях точки $\Omega$ от $s_1,s_2$ мы должны заменить
радикалы (\ref{eq4_28}) на
\begin{equation}\label{eq4_32}
P^*_\gamma   = R_{1\gamma } R^*_{2\gamma }.
\end{equation}
Обозначим
\begin{equation}\notag
y_\gamma = \lsgn P^*_\gamma = u_\gamma \oplus u_{5+\gamma} \quad
(\gamma=1\div 3).
\end{equation}
Представим зависимости (\ref{eq4_27}) в виде
\begin{equation}\label{eq4_33}
\omega_1=\frac{a'_0+a'_1 P^*_1 P^*_3+a'_2 P^*_2 P^*_3}{a_0+a_1 P^*_1
P^*_3+a_2 P^*_2 P^*_3}, \qquad \omega_2=\frac{a'' R^*_{25}
P^*_3}{a_0+a_1 P^*_1 P^*_3+a_2 P^*_2 P^*_3}.
\end{equation}
Здесь $a_i,a'_i,a''$ -- однозначные функции от $s_1,s_2$. Числитель
и знаменатель обоих выражений домножены на $R^*_{25} P^*_3$, с тем
чтобы получить однозначные зависимости точек $\Omega$ от комбинаций
радикалов. Для описания надстройки ${(s_1,s_2)\mapsto \Omega}$
вводим АБВФ вида
\begin{eqnarray}
& {C}(u_1,...,u_{10})=(z_1,z_2,z_3), \notag \\
& z_1=y_1 \oplus y_3, \quad z_2=y_2 \oplus y_3, \quad
z_3=u_{10}\oplus y_3. \label{eq4_34}
\end{eqnarray}
Из леммы~\ref{lem4} сразу же следует, что, поскольку аргумент
$u_{10}$ всегда относится ко второй группе, значение $z_3$ вместе с
этим аргументом можно отбросить. По лемме~\ref{lem1} можно также не
учитывать фиктивные аргументы $u_4,u_5,u_9$. Таким образом, остается
функция ${C}: \mB^6 \to \mB^2$. В свою очередь, она представляется в
виде ${C}={B}\circ {D}$, где
\begin{equation}\label{eq4_35}
{D}(\mbf{u})=(y_1,y_2,y_3), \qquad {B}(y_1,y_2,y_3)=(y_1 \oplus y_3,
y_2 \oplus y_3),
\end{equation}
и так определенная композиция удовлетворяет всем условиям леммы
\ref{lem7}. Поэтому на самом деле достаточно проанализировать классы
отображения ${B}:\mB^3 \to \mB^2$ всего с тремя булевыми
аргументами, относя их к необходимой группе. Как следует из
табл.~\ref{tabkow1}, в области \ts{III} ко второй группе относятся
аргументы $y_1,y_3$, в области \ts{V} эта группа включает только
$y_1$, а в остальных областях -- только $y_3$.

В областях \ts{I}, \ts{II} условие $P^*_1 P^*_2 >0$ принимает вид
$y_1\oplus y_2=0$, следовательно, пара $(y_1,y_2)$ может принимать
всего два значения $00,11$. Но ${B}(000)={{B}(111)=00}$. Поэтому в
этих областях класс эквивалентности один. В области \ts{III} по
лемме~\ref{lem4} можно вначале отбросить пару $z_1,y_1$, после чего
условиям леммы удовлетворяет и пара $z_2,y_3$. По
замечанию~\ref{rem2} класс эквивалентности один. В области~\ts{IV}
аргументы $y_1,y_2$ образуют два класса $\{00,11\},\{01,10\}$. В
области \ts{V} такие же классы образуют аргументы $y_2,y_3$, но
здесь два класса получаются сразу же исключением пары $z_1,y_1$ по
лемме \ref{lem4}, после чего остаются два аргумента первой группы и
одна зависимая переменная $y_2\oplus y_3$, два значения которой,
очевидно, и порождают разные классы. Итак, множество
$\A_\omega(\mbf{f})$ связно в областях \ts{I} -- \ts{III}, и имеет
две компоненты в областях \ts{IV}, \ts{V}. Все варианты показаны на
рис.~\ref{fig_kow2}, где видно, как происходит разворачивание
полуполосы на плоскости $(s_1,s_2)$ в множество
$\A_\omega(\mbf{f})$. По количеству наборов аргументов булевой
вектор-функции ${B}$, заданной в (\ref{eq4_35}), множество
$\A_\omega(\mbf{f})$ формируется из четырех <<лоскутов>> в областях
\ts{I} -- \ts{II}, и из восьми -- в остальных случаях. Отметим, что
ось $\omega_2=0$ всегда разделяет <<лоскуты>> ввиду присутствия в
числителе соответствующего выражения (\ref{eq4_33}) меняющего знак
радикала $R^*_{25}$. Естественно, имеет место полное совпадение с
результатами \cite{Zhuk}, а рис.~\ref{fig_kow2} можно найти,
например, в \cite{Appel}. Однако нам не пришлось прибегать ни к
каким аналитическим выкладкам, повторить которые за классиками
практически невозможно.

\begin{figure}[ht]
\centering
\includegraphics[width=100mm,keepaspectratio]{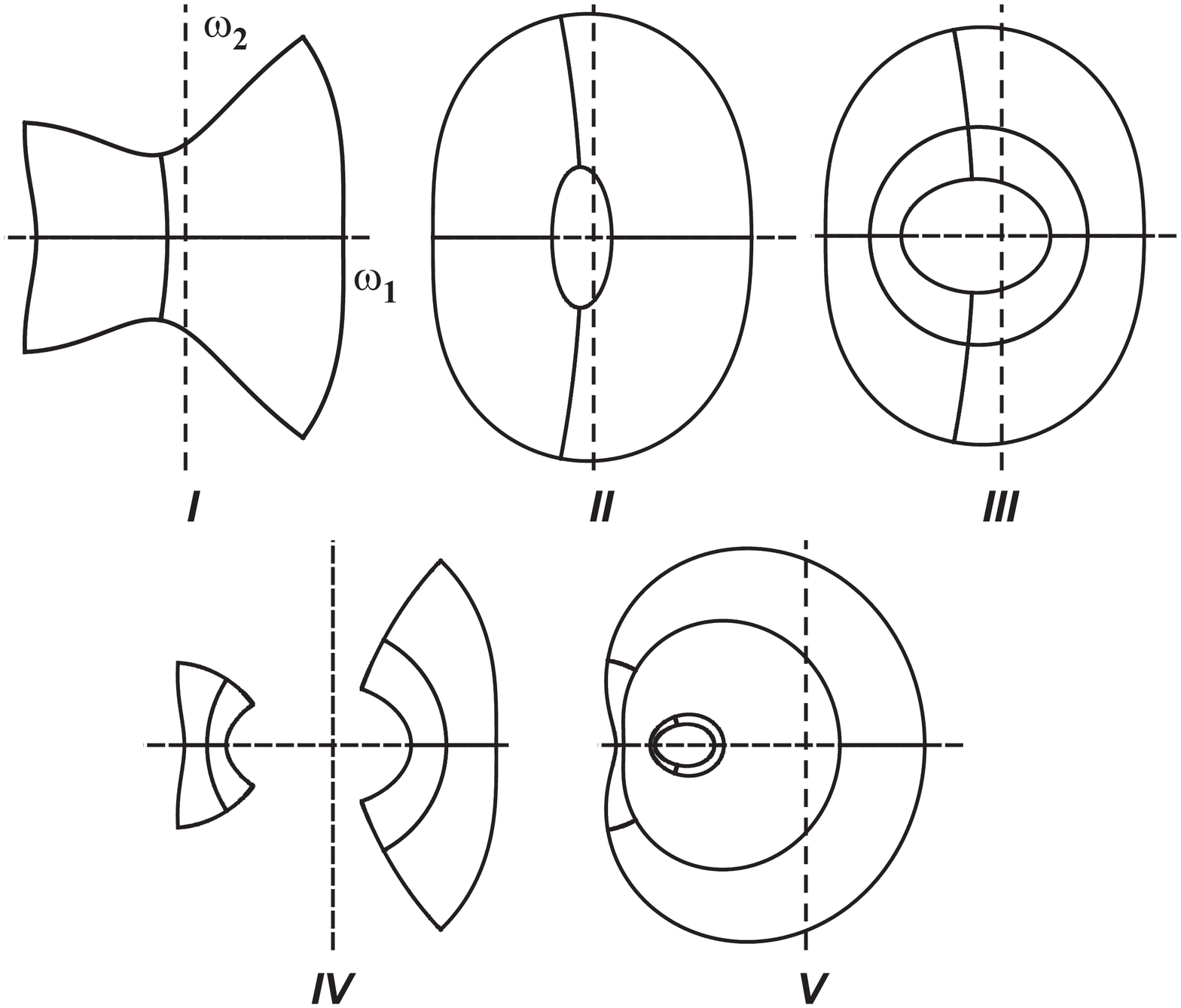}
\caption{Образы $\mF_\mbf{f}$ на плоскости
$(\omega_1,\omega_2)$.}\label{fig_kow2}
\end{figure}

\subsection{Случай Ковалевской. Интегральные многообразия}
Обратимся к задаче определения количества связных компонент в
составе регулярных интегральных многообразий $\mF_\mbf{f}$. Для
полного описания надстройки необходимо к формулам (\ref{eq4_27})
добавить выражения для $\omega_3,\alpha_3$ \cite{Kowa,Appel}:
\begin{equation}\label{eq4_36}
\omega _3  = \sqrt{2} \frac{ \sum
\displaystyle{\frac{\sqrt{e_\gamma}}{\varphi'(e_\gamma)}}P_{\delta
\nu} }{ \sum
\displaystyle{\frac{\sqrt{e_\gamma}}{\varphi'(e_\gamma)}}P_{
\gamma}}, \qquad \alpha_3  = \frac{ \sum
\displaystyle{\frac{\sqrt{e_\delta}\sqrt{e_\nu}}{\varphi'(e_\gamma)}}P_{\delta
\nu} } {\sum
\displaystyle{\frac{\sqrt{e_\gamma}}{\varphi'(e_\gamma)}}P_{\delta
\nu}},
\end{equation}
где
\begin{equation}\notag
\begin{array}{l}
P_{\delta
\nu}=\displaystyle{\frac{P_{\delta}P_{\nu}}{s_1-s_2}}\left\{\displaystyle{\frac{S_1}{(s_1-e_\delta)(s_1-e_\nu)}}
-\displaystyle{\frac{S_2}{(s_2-e_\delta)(s_2-e_\nu)}
}\right\},\\[5mm]
S_1=\sqrt{S(s_1)}, \qquad S_2=\sqrt{S(s_2)}.
\end{array}
\end{equation}

Чтобы использовать уже введенные функции (\ref{eq4_34}), представим
эти выражения в виде
\begin{equation}\notag
\begin{array}{l}
\omega_3=\displaystyle{\frac{\sum b_\gamma P^*_{\delta\nu}}{a_0+a_1
P^*_1 P^*_3+a_2 P^*_2 P^*_3}},\qquad
\alpha_3=\displaystyle{\frac{\sum b'_\gamma P^*_{\delta\nu}}{\sum
b_\gamma P^*_{\delta\nu}}},
\end{array}
\end{equation}
где $a_i $ ($i = 0 \div 2$), $b_\gamma  ,b'_\gamma  $ ($\gamma  = 1
\div 3$) -- однозначные функции переменных $s_1 ,s_2 $. Величины
$a_i $ определены в (\ref{eq4_33}), а $P_{\delta \nu }^* $ получены
из $P_{\delta \nu }^{} $ умножением на $R_{25}^* P_3^* $. Каждое из
$P_{\delta \nu }^* $ представляет собой сумму двух слагаемых
\begin{equation}\label{eq4_37}
c_1 R_{25}^* P_3^* P_\delta ^* P_\nu ^* S_1  + c_2 R_{25}^* P_3^*
P_\delta ^* P_\nu ^* S_2^* .
\end{equation}
Здесь использованы обозначения (\ref{eq4_30}), (\ref{eq4_32}), в
соответствии с которыми
\begin{equation}\notag
S_1  = \prod\limits_{\gamma  = 1}^5 {R_{1\gamma } } ,\quad S_2^*  =
\prod\limits_{\gamma  = 1}^5 {R_{2\gamma }^* } = (R^*_{25})^6 S_2,
\end{equation}
а $c_1 ,c_2 $ -- однозначные функции переменных $s_1 ,s_2$. Поэтому
в дополнение к (\ref{eq4_34}) вводим компоненты АБВФ, отвечающие за
знаки радикалов в (\ref{eq4_37}):
\begin{equation}\notag
\begin{array}{l}
z_4  = \lsgn (R_{25}^* P_2^* S_1 ) = u_1  \oplus u_3  \oplus u_4  \oplus u_5  \oplus u_7  \oplus u_{10} ,  \\
z_5  = \lsgn (R_{25}^* P_1^* S_1 ) = u_2  \oplus u_3  \oplus u_4  \oplus u_5  \oplus u_6  \oplus u_{10} ,  \\
z_6  = \lsgn (R_{25}^* P_1^* P_2^* P_3^* S_1 ) = u_4  \oplus u_5  \oplus u_6  \oplus u_7  \oplus u_8  \oplus u_{10} ,  \\
z_7  = \lsgn (R_{25}^* P_2^* S_2^* ) = u_2  \oplus u_6  \oplus u_8  \oplus u_9 ,  \\
z_8  = \lsgn (R_{25}^* P_1^* S_2^* ) = u_1  \oplus u_7  \oplus u_8  \oplus u_9 ,  \\
z_9  = \lsgn (R_{25}^* P_1^* P_2^* P_3^* S_2^* ) = u_1  \oplus u_2
\oplus u_3  \oplus u_9 .
\end{array}
\end{equation}
Матрица построенной булевой вектор-функции ${C} : \mB^{10} \to
\mB^9$ такова (нуль изображаем пустой клеткой, чтобы не загромождать
матрицу):
\begin{equation}\notag
\begin{array}{c||c  c  c  c c c c c c c ||} \multicolumn{1}{c}{} &
\multicolumn{1}{c}{u_1 } & \multicolumn{1}{c}{u_2 } &
\multicolumn{1}{c}{u_3 } & \multicolumn{1}{c}{u_4 } &
\multicolumn{1}{c}{u_5 } & \multicolumn{1}{c}{u_6 } &
\multicolumn{1}{c}{u_7 } & \multicolumn{1}{c}{u_8 } &
\multicolumn{1}{c}{u_9 } &
\multicolumn{1}{c}{u_{10} } \\
{z_1} & {1} & { } & {1} & { }   & { }  & {1} & { } & {1} & { } & { } \\
{z_2} & { } & {1} & {1} & { }   & { }  & { } & {1} & {1} & { } & { } \\
{z_3} & { } & { } & {1} & { }   & { }  & { } & { } & {1} & { } & {1} \\
{z_4} & {1} & { } & {1} & {1}   & {1}  & { } & {1} & { } & { } & {1} \\
{z_5} & { } & {1} & {1} & {1}   & {1}  & {1} & { } & { } & { } & {1} \\
{z_6} & { } & { } & { } & {1}   & {1}  & {1} & {1} & {1} & { } & {1} \\
{z_7} & { } & {1} & { } & { }   & { }  & {1} & { } & {1} & {1} & { } \\
{z_8} & {1} & { } & { } & { }   & { }  & { } & {1} & {1} & {1} & { } \\
{z_9} & {1} & {1} & {1} & { }   & { }  & { } & { } & { } & {1} & { }
\end{array}\,.
\end{equation}
Далее осуществляем процедуру редукции эквивалентными
преобразованиями по леммам~\ref{lem4},~\ref{lem5}. Заметим (см.
табл.~\ref{tabkow1}), что аргумент $u_5 $ всегда относится ко второй
группе. Добьемся того, что он останется в единственной строке. Для
этого преобразуем строки $z_5 ,z_6 $ в соответствии с
леммой~\ref{lem5}: $z_5  \to z_5  \oplus z_4 $, $z_6  \to z_6 \oplus
z_4$. Теперь $u_5$ входит только в $z_4$, и по лемме \ref{lem4} эти
столбец и строку можно исключить, не меняя классов эквивалентности.
Аргумент $u_4$ стал фиктивным и подлежит исключению. Аргумент
$u_{10}$ также всегда относится ко второй группе и, после
выполненных преобразований, входит только в строку $z_3$, поэтому
$z_3$ и $u_{10}$ исключаем. Замены $z_8  \to z_8  \oplus z_5 \oplus
z_7 $, $z_9  \to z_9 \oplus z_6  \oplus z_7 $, $z_6  \to z_6 \oplus
z_1 $, $z_5  \to z_5 \oplus z_1  \oplus z_2 $  делают константами
компоненты $z_8 ,z_9 ,z_5 ,z_6$. Исключаем их. Для наглядности
выполним еще подстановки $z_1 \to z_1 \oplus z_7 $, $z_2 \to z_2
\oplus z_7 $. В результате получим матрицу редуцированной АБВФ:
\begin{equation}\label{eq4_38}
\begin{array}{c||c  c  c  c c c c ||} \multicolumn{1}{c}{} &
\multicolumn{1}{c}{u_1 } & \multicolumn{1}{c}{u_2 } &
\multicolumn{1}{c}{u_3 } & \multicolumn{1}{c}{u_6 } &
\multicolumn{1}{c}{u_7 } & \multicolumn{1}{c}{u_8 } &
\multicolumn{1}{c}{u_9 } \\
{z_1} & {1} & { } & { } & { }   & {1}  & {1} & {1} \\
{z_2} & { } & {1} & { } & {1}   & { }  & {1} & {1} \\
{z_7} & { } & { } & {1} & {1}   & {1}  & { } & {1}
\end{array}\,.
\end{equation}
Ясно, что ввиду наличия единичного блока по аргументам $u_1, u_2,
u_3$ и того факта, что ни один из этих аргументов не может всегда
относиться ко второй группе, дальнейшая редукция одновременно для
всех областей, указанных в табл.~\ref{tabkow1}, невозможна. Поэтому
выполним преобразования в зависимости от области, которой
принадлежат постоянные интегрирования.

В областях \ts{I}, \ts{III}, \ts{IV} аргумент $u_3$ относится ко
второй группе. Исключаем $z_7 ,u_3$. Получим матрицу
\begin{equation}\notag
\begin{array}{c||c  c  c c c c ||} \multicolumn{1}{c}{} &
\multicolumn{1}{c}{u_1 } & \multicolumn{1}{c}{u_2 } &
\multicolumn{1}{c}{u_6 } & \multicolumn{1}{c}{u_7 } &
\multicolumn{1}{c}{u_8 } &
\multicolumn{1}{c}{u_9 } \\
{z_1} & {1} & { } & { } & {1}   & {1}  & {1} \\
{z_2} & { } & {1} & {1} & { }   & {1}  & {1}
\end{array}\,.
\end{equation}
В области \ts{III} аргумент $u_6$ относится ко второй группе,
поэтому здесь исключаем $z_2 ,u_6$. Остается одна компонента
\begin{equation}\label{eq4_39}
z_1  = u_1  \oplus u_7  \oplus u_8  \oplus u_9 .
\end{equation}
В областях \ts{I}, \ts{IV} ко второй группе относится аргумент
$u_9$, присутствующий в обеих компонентах $z_1 ,z_2 $. Поэтому
выполним подстановку $z_1  \to z_1  \oplus z_2 $, после чего
исключим $u_9,z_2$. Остается одна компонента
\begin{equation}\label{eq4_40}
z_1  = u_1  \oplus u_2  \oplus u_6  \oplus u_7.
\end{equation}
Все оставшиеся аргументы в (\ref{eq4_39}), (\ref{eq4_40}) относятся
к первой группе. В областях \ts{III}, \ts{IV} ограничений на
аргументы нет, поэтому оставшаяся компонента принимает оба значения
и они порождают разные классы эквивалентности. В области \ts{I}
корни $e_1,e_2$ комплексно сопряжены, поэтому имеют место
ограничения (см. замечание~\ref{rem1}) $u_1  \oplus u_2 = 0$, $u_6
\oplus u_7 = 0$ и функция (\ref{eq4_40}) принимает только одно
значение, так что класс эквивалентности один.

В области \ts{II} ко второй группе относится аргумент $u_8$, а в
области \ts{V} -- аргумент $u_6$. Выполним поэтому в матрице
(\ref{eq4_38}) следующие подстановки: вначале $z_2 \to z_2 \oplus
z_1 $, затем $z_7  \to z_7  \oplus z_2$. Получим
\begin{equation}\label{eq4_41}
\begin{array}{c||c  c  c  c c c c ||} \multicolumn{1}{c}{} &
\multicolumn{1}{c}{u_1 } & \multicolumn{1}{c}{u_2 } &
\multicolumn{1}{c}{u_3 } & \multicolumn{1}{c}{u_6 } &
\multicolumn{1}{c}{u_7 } & \multicolumn{1}{c}{u_8 } &
\multicolumn{1}{c}{u_9 } \\
{z_1} & {1} & { } & { } & { }   & {1}  & {1} & {1} \\
{z_2} & {1} & {1} & { } & {1}   & {1}  & { } & { } \\
{z_7} & {1} & {1} & {1} & { }   & { }  & { } & {1}
\end{array}\,.
\end{equation}
Теперь в области \ts{II} исключаем $z_1 ,u_8$. Компонента $z_2  =
u_1  \oplus u_2  \oplus u_6  \oplus u_7 $. Но здесь также имеют
место ограничения (\ref{eq4_31}), поэтому эта компонента фиктивна, а
для последней оставшейся компоненты получаем $z_7  = u_3 \oplus
u_9$. Оба аргумента из первой группы, ограничениями не связаны,
поэтому классов эквивалентности два. В области \ts{V} из
(\ref{eq4_41}) исключаем $z_2 ,u_6$. Получаем матрицу редуцированной
АБВФ
\begin{equation}\notag
\begin{array}{c||c  c  c  c c c ||} \multicolumn{1}{c}{} &
\multicolumn{1}{c}{u_1 } & \multicolumn{1}{c}{u_2 } &
\multicolumn{1}{c}{u_3 } & \multicolumn{1}{c}{u_7 } &
\multicolumn{1}{c}{u_8 } &
\multicolumn{1}{c}{u_9 } \\
{z_1} & {1} & { } & { } & {1}   & {1}  &  {1} \\
{z_7} & {1} & {1} & {1} & { }   & { }  &  {1}
\end{array}\,.
\end{equation}
Все аргументы из первой группы, ограничениями не связаны, поэтому
все наборы значений  пары $(z_1 ,z_7)$ порождают различные классы и
классов эквивалентности четыре.

Получено строгое доказательство теоремы \cite{Kh831} о количестве
торов Лиувилля.
\begin{theorem}
В случае интегрируемости Ковалевской регулярные многообразия
$\mF_\mbf{f}$ имеют следующий тип: ${\mbf{T}}^2 $ в области \ts{I},
$2{\mbf{T}}^2 $ в областях \ts{II}, \ts{III}, \ts{IV} и
$4{\mbf{T}}^2 $ в области \ts{V}.
\end{theorem}

\begin{figure}[ht]
\centering
\includegraphics[width=70mm,keepaspectratio]{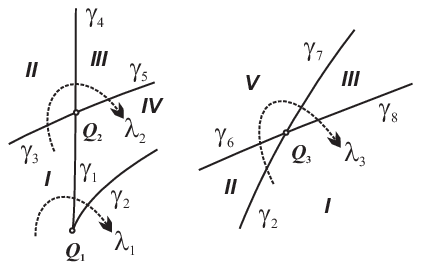}
\caption{Некоторые особые точки.}\label{fig_kow3}
\end{figure}

Как отмечалось в общей теории, критические интегральные поверхности
устанавливаются аналогично. Поскольку все результаты полностью
совпадают с описанием перестроек, данных в \cite{Kh831}, ограничимся
рассмотрением окрестностей трех особых точек, обозначенных на
рис.~\ref{fig_kow1} как $Q_1 - Q_3$. В увеличенном виде эти
окрестности приведены на рис.~\ref{fig_kow3}. Рассмотрим обходы
указанных точек по путям, обозначенным $\lambda_1 - \lambda_3$. При
этом пересекаются кривые -- участки разделяющих кривых
$\pi_3,\pi_5,\pi_6$, обозначенные через $\gamma_1 - \gamma_8$.
Теперь уже установлено, что все они принадлежат бифуркационной
диаграмме (так как интегральные поверхности не пусты). Необходимая
информация на этих участках собрана в табл.~\ref{tabkow2}. В
столбцах~4 и 5 через $e_*$ обозначен кратный корень, значение
которого ясно из столбца~3.
\begin{table}[ht]
\centering
\begin{tabular}{|c|c|c|c|c|}
\multicolumn{5}{r}{\fts{Таблица \myt\label{tabkow2}}}\\
\hline
\begin{tabular}{c}Переход,\\кривая\end{tabular} &\begin{tabular}{c}Корни\end{tabular}
&\begin{tabular}{c}Область\\изменения $s_1$\end{tabular}
&\begin{tabular}{c}Область\\изменения $s_2$\end{tabular}&\begin{tabular}{c}Первая\\группа\end{tabular} \\
\hline \ts{I}-\ts{IV}, $\gamma_1$ &$e_4<e_1<e_2=e_3<e_5$ &
$[\,e_*,e_5\,]$ & $[-\infty,e_4\,]$ &
14678\\
\hline \ts{IV}-\ts{I}, $\gamma_2$& $e_4<e_1=e_2<e_3<e_5$ &
$[\,e_3,e_5\,]$ & $[-\infty,e_4\,]$ &
124678\\
\hline \ts{I}-\ts{II}, $\gamma_3$&$e_3=e_4<e_5$ & $[\,e_*,e_5\,]$ &
$[-\infty,e_*\,]$ &
$1267^*$\\
\hline \ts{II}-\ts{III}, $\gamma_4$&$e_1<e_4<e_2=e_3<e_5$ &
$[\,e_*,e_5\,]$ & $[-\infty,e_1\,]$ &
14789\\
\hline \ts{III}-\ts{IV}, $\gamma_5$&$e_4=e_1<e_2<e_3<e_5$ &
$[\,e_3,e_5\,]$ & $[-\infty,e_*\,]$ &
12478\\
\hline \ts{II}-\ts{V}, $\gamma_6$&$e_1=e_2<e_3<e_4<e_5$ &
$[\,e_4,e_5\,]$ & $[-\infty,e_*\,]$ &
$12389^{(*)}$\\
\hline \ts{V}-\ts{III}, $\gamma_7$&$e_1<e_2<e_3=e_4<e_5$ &
$[\,e_*,e_5\,]$ & $[-\infty,e_1\,]$ &
12789\\
\hline \ts{III}-\ts{I}, $\gamma_8$&$e_1=e_2<e_4<e_3<e_5$ &
$[\,e_3,e_5\,]$ & $[-\infty,e_*\,]$ &
$12489^{(*)}$\\
\hline
\end{tabular}
\end{table}
Последний столбец требует пояснений. На кривой $\gamma_1$ кратный
корень лежит на границе изменения $s_1$. При этом в реальном времени
либо $s_1 \equiv e_*$, либо $s_1 \to e_*$ при $t \to \pm \infty$. Но
в топологическом аспекте мы рассматриваем надстройки над отрезком
$[\,e_*,e_5\,]$ в виде петель, поэтому считаем $t=\infty$ обычным
значением. В силу этого радикалы $R_{12},R_{13}$ считаем меняющими
знак. На кривой $\gamma_2$ кратный корень возникает за пределами
изменения переменных $s_1,s_2$. Иногда это означает (как, например,
при переходе между областями \ts{III-a} и \ts{III-b}), что
бифуркации не происходит. Чтобы это проверить, нужно избавиться от
устранимой особенности в выражениях (\ref{eq4_27}), (\ref{eq4_36}),
связанной с тем, что $\varphi'(e_*)=0$. При этом надо рассматривать
два варианта выбора знаков радикалов. В данном случае оказывается,
что на одном из торов Лиувилля одна из образующих стягивается в
точку. Это аналитически выражается в следующем: в одном из вариантов
предельного перехода $\omega_2$ обращается в тождественный нуль, а
$\omega_1$ в константу. Все сказанное {\it не влияет} на количество
связных компонент интегральной поверхности, и исследование
соответствующих деталей здесь просто не нужно. Существенно же
следующее. Радикалы $R_{11},R_{12}$, равно как и $R_{21},R_{22}$,
формально получают одинаковое выражение $\sqrt{s_i-e_*}$. Однако при
этом знаки в каждой паре можно выбирать независимо четырьмя
способами, в связи с чем все четыре радикала отнесены к первой
группе. На кривой $\gamma_3$ звездочка в последнем столбце означает,
как и ранее, дополнительные условия (\ref{eq4_31}). Аналогично
устанавливаются радикалы первой группы на кривых $\gamma_4$,
$\gamma_5$, $\gamma_7$. На кривых $\gamma_6$, $\gamma_8$ кратный
корень включен в область изменения $s_2$, поэтому радикалы
$R_{21},R_{22}$ необходимо считать меняющими знак (в точке
$t=\infty$). Но для переменной $s_1$ соответствующие радикалы
$R_{11},R_{12}$ можно рассматривать двояко. При выходе на
разделяющую кривую из областей \ts{III} и \ts{V} они возникли из
радикалов, не меняющих знак, поэтому должны унаследовать четыре
набора знаков. Но при выходе на те же участки из областей \ts{I} и
\ts{II} они возникают на месте комплексно сопряженной пары, поэтому
нужно считать, что $R_{21} R_{22}>0$. Такая двойственность в
рассуждениях, конечно же, недопустима, и на участках $\gamma_6$,
$\gamma_8$ мы увидим, как алгоритм редукции АБВФ уничтожает эту
неопределенность. Пока же столбец, описывающий первую группу, для
этих случаев помечен звездочкой в скобках, то есть звездочка здесь
не обязательна -- результат от этого не зависит. Применяем описанную
ранее технику. Все действия основаны на леммах~\mbox{\ref{lem4},
\ref{lem5}} с учетом замечания~\ref{rem2}. Исходим во всех случаях
из матрицы (\ref{eq4_38}), при построении которой не требовался учет
конкретного распределения корней $e_\gamma$.

На $\gamma_1$ последовательно исключаем пары $(u_2,z_2)$,
$(u_3,z_7)$, $(u_7,z_1)$. Исключены все компоненты функции. По
замечанию~\ref{rem2} класс эквивалентности один.

На $\gamma_2$, $\gamma_3$ выполняем те же действия, что и в областях
\ts{II}, \ts{IV}: исключаем $(u_3,z_7)$, подставляем $z_1 \to z_1
\oplus z_2$, исключаем $(u_9,z_2)$. Остается компонента
(\ref{eq4_40}), зависящая только от аргументов первой группы. На
$\gamma_2$ классов эквивалентности два, а на $\gamma_3$ оставшаяся
компонента -- константа, поэтому класс эквивалентности один.

На $\gamma_4$ исключаем $(u_2,z_2)$, затем $(u_3,z_7)$. Остается
компонента (\ref{eq4_39}). Классов эквивалентности два.

На $\gamma_5$ исключаем $(u_3,z_7)$, затем $(u_6,z_2)$ и, наконец,
$(u_9,z_1)$. Класс эквивалентности один.

На $\gamma_6$ исключаем $(u_6,z_2)$ и $(u_7,z_1)$. Остается
$z_7=(u_1\oplus u_2)\oplus u_3\oplus u_9$. Независимо от того,
имеется или нет условие $u_1 \oplus u_2=0$, аргументы $u_3,u_9$ из
первой группы дают два значения $z_7$, порождающих два различных
класса эквивалентности.

На $\gamma_7$ исключаем $(u_3,z_7)$, $(u_6,z_2)$. Остается
$z_1=u_1\oplus u_7\oplus u_8\oplus u_9$. Классов эквивалентности
два.

На $\gamma_8$ последовательно исключаем пары $(u_3,z_7)$,
$(u_6,z_2)$, $(u_7,z_1)$. По замечанию~\ref{rem2} класс
эквивалентности один. Вопрос о необходимости применения ограничений
отпадает.

Теперь перестройки вдоль путей $\lambda_1-\lambda_3$, приведенных на
рис.~\ref{fig_kow3}, выписываются однозначно:
\begin{equation}\notag
\begin{array}{l}
\lambda_1: \bT^2 \to (S^1 \vee S^1) \times S^1 \to 2\bT^2 \to \bT^2
\cup S^1
\to \bT^2; \\[3mm]
\lambda_2: \bT^2 \to (S^1 \vee S^1) \times S^1 \to 2\bT^2 \to 2(S^1
\vee S^1) * S^1 \to
\\
\qquad \to 2\bT^2 \to (S^1 \ddot \vee S^1)\times S^1 \to 2\bT^2;\\[3mm]
\lambda_3: \bT^2 \to (S^1 \vee S^1) \times S^1 \to 2\bT^2 \to 2(S^1
\vee S^1) \times S^1
\to\\
\qquad \to 4 \bT^2 \to 2(S^1 \vee S^1) \times S^1 \to 2\bT^2 \to (S^1 \vee S^1) \times S^1 \to \bT^2.\\
\end{array}
\end{equation}

На этом топологический анализ задачи Ковалевской закончен. Отметим
еще раз, что все эти результаты ранее доказаны в \cite{Kh831, KhBk}
достаточно сложными вычислениями, которые некоторыми современными
авторами называются <<утомительными>> \cite{Oden}. Здесь же задача
сведена к простым действиям с двоичными матрицами.

В следующей части статьи будут приведены исследования новых случаев
существования алгебраических решений для критических подсистем
волчка Ковалевской в двойном поле.

{\bf Благодарности.} Работа выполнена при финансовой поддержке
грантов РФФИ № 10-01-00043 (применение булевых функций и их
редукция), РФФИ и Администрации Волгоградской области № 10-01-97001
(разработка и реализация вычислительных алгоритмов).

\end{document}